\documentclass[12pt,a4paper]{amsart}

\usepackage{latexsym,amssymb,amsmath,amsthm,amsfonts,enumerate,verbatim,xspace,
exscale}
\usepackage{graphicx}
\usepackage{color,amsbsy}
\usepackage{tikz}
\usetikzlibrary{patterns}

\usepackage{hyperref}
\definecolor{darkblue}{rgb}{0,0,.5}
\hypersetup{colorlinks=true, breaklinks=true, linkcolor=darkblue, menucolor=darkblue,  urlcolor=darkblue,citecolor=darkblue}

\usepackage{setspace}

\theoremstyle{plain}
\newtheorem{theorem}{Theorem}[section]
\newtheorem{lemma}[theorem]{Lemma}
\newtheorem{proposition}[theorem]{Proposition}

\theoremstyle{definition}

\newtheorem{remark}[theorem]{Remark}

\def\d{\textup{div}}
\def\D{\mathcal{D}}

\def\R{\mathbb{R}}
\def\I{\mathbb{I}}

\def\M{\mathcal{M}}

\def\F{\mathcal{F}}

\def\A{\mathcal{A}}

\newcommand{\up}{\upshape}

\def\vv<#1>{\langle#1\rangle}
\def\ww<#1>{\langle\langle#1\rangle\rangle}

\newcommand{\diag}[1]{\mbox{$\textup{diag}(#1)$}}

\newcommand{\pr}{\mbox{$\text{\up{pr}}$}}

\providecommand{\vol}{\mbox{$\text{\up{vol}}$}}

\newcommand{\Ver}{\mbox{$\textup{Ver}$}}
\newcommand{\Hor}{\mbox{$\textup{Hor}$}}

\newcommand{\dd}[2]{\mbox{$\frac{\partial #2}{\partial #1}$}}

\providecommand{\del}{\partial}

\newcommand{\om}{\omega}
\newcommand{\Om}{\Omega}
\newcommand{\var}{\varphi}

\newcommand{\eps}{\varepsilon}
\newcommand{\lam}{\lambda}

\newcommand{\by}[2]{\mbox{$\frac{#1}{#2}$}}

\providecommand{\set}[1]{\mbox{$\{#1\}$}}

\newcommand{\ao}{\textup{aut}_0}
\newcommand{\du}{\textup{diff}_0} 
\newcommand{\gu}{\mathfrak{g}}

\newcommand{\ko}{\mathfrak{k}}

\newcommand{\Ad}{\mbox{$\text{\upshape{Ad}}$}}
\newcommand{\ad}{\mbox{$\text{\upshape{ad}}$}}

\newcommand{\SO}{\mbox{$\textup{SO}$}}
\newcommand{\so}{\mbox{$\mathfrak{so}$}}

\newcommand{\X}{\mbox{$\mathcal{X}$}}

\newcommand{\aut}{\mbox{$\textup{aut}$}}

\newcommand{\diff}{\mbox{$\textup{diff}$}}

\newcommand{\gau}{\mbox{$\textup{gau}$}}
\newcommand{\G}{\mbox{$\mathcal{G}$}}

\usepackage[all]{xy}

\title[Feedback control of charged ideal fluids]{%
Feedback control of charged ideal fluids}

\author{Simon Hochgerner}
\address{Finanzmarktaufsicht (FMA),
Otto-Wagner Platz 5, A-1090 Vienna
}
\email{simon.hochgerner@fma.gv.at} 

\begin{document}

\begin{abstract}
The theory of controlled mechanical systems of \cite{K85,BKMS92,BMS97,BLM01a,BLM01b} is extended to the case of ideal incompressible fluids consisting of charged particles in the presence of an external magnetic field.  
The resulting control is of feedback type and depends on the Eulerian state of the controlled system.  Moreover, the control is set up so that the corresponding closed loop equations are Lie-Poisson. This implies that the energy-momentum method of \cite{A66,HMRW85} can be used to find a stabilizing control.

As an example the case of planar parallel shear flow with an inflection point is treated. A state dependent feedback control is constructed which stabilizes the system for an arbitrarily long channel.
\end{abstract}

\maketitle

\section{Introduction}

\subsection{Plasma dynamics and stability problem}
A plasma is a collection of charged particles and is often described as a fluid mechanical system. 
The motion of the plasma fluid generates an electromagnetic field and, in turn, the electromagnetic field acts on the fluid through the Lorentz force.  
There exists  a hierarchy of (Hamiltonian) descriptions of plasma dynamics, including the Vlasov-Maxwell equations, the Navier-Stokes-Maxwell equations, magnetohydrodynamics, or the Euler-Poisson equations. 
The Navier-Stokes-Maxwell equations arise by coupling the Navier-Stokes to the Maxwell equations through the Lorentz force and the current density. 
If viscous effects can be neglected one obtains the Euler-Maxwell equations.

An important topic in plasma physics, and in particular in nuclear fusion research, is the study of stability of equilibria. Specifically, one is interested in controlling an externally generated electromagnetic field in such a way that the motion of the plasma fluid is given by a stable equilibrium. 
See  \cite{C84,Davidson,Lin69,MWRSS83,CGG00,GM14}.

This paper studies the stability control problem for a system which mimics the Euler-Maxwell system in the sense that the Euler equation is kept but instead of a self-consistent electromagnetic field subject to the Maxwell equations one considers only a constant magnetic field. Further, it is assumed that the fluid is incompressible and couples to the magnetic field $B$ by a charge $q$ that is carried along flow lines. The resulting system is described by the charged Euler equations (\cite{GBTV13}), which in $\R^3$ take the form 
\begin{align}
\label{1e:EM_u_d}
    \dot{u} + \nabla_{u}u
    &=  - \nabla p
        + q \, u\times B\\
\label{1e:div0}
    \textup{div}\,u
    &= 0 \\
\label{1e:cont_d}
    \dot{q} + \vv<u,\nabla q>
    &= 0
\end{align}
where $p$ is the pressure, determined by the requirement that  $\dot{u}$ in \eqref{1e:EM_u_d} is divergence free.  
While this system is a strong simplification of the Euler-Maxwell equations, two essential features are retained: The fluid velocity $u$ evolves according to the Euler equation and the charge transport~\eqref{1e:cont_d} follows from the Noether Theorem for the action of the gauge symmetry group. 

The observation of this paper is that the theory of \cite{K85,BKMS92,WK92,BMS97,BLM01a,BLM01b} (see Section~\ref{sec:intro-fb}) can be adapted to find a characterization of controls which stabilize a given \emph{previously unstable} equilibrium of 
the charged Euler system. The resulting stabilizing control, if it exists, yields stability in the nonlinear sense (following \cite{A66,HMRW85}) with respect to perturbations in the fluid's velocity.   

The controlled quantity in this approach is the charge $q$ and the control should be of feedback type with respect to the state of the system. Thus \eqref{1e:cont_d} is replaced by a control law $q = \mathcal{U}(u)$ which depends on observations of the fluid's velocity. 
The mathematical formulation of this procedure is described in Section~\ref{sec:introLP}. 
Physically, this means that the charge of the fluid is directly controlled at all times and each point in space. 
Due to the external magnetic field this influences the fluid motion according to \eqref{1e:EM_u_d}. An example where such a control can be explicitly found and leads to a stabilization of a (previously unstable) fluid equilibrium is the shear flow system in Section~\ref{sec:shear}.

The Euler-Maxwell equations can be derived from a Hamiltonian system, and the same is true for the charged Euler equations. However, the latter can be viewed, moreover, as a Lie-Poisson system (see Appendix~\ref{app:GM}). This is a structural simplification and makes the application of the Controlled Hamiltonian ideas \cite{K85,BKMS92,WK92,BMS97,BLM01a,BLM01b} (see Section~\ref{sec:intro-fb}) easier. But there is no a priori reason why the Controlled Hamiltonian approach should not also work for the full Euler-Maxwell system. Then the charge control would be realized by an externally impressed electric current. 

\subsection{Feedback control of mechanical systems with symmetries}\label{sec:intro-fb}
The feedback control method of Lagrangian or Hamiltonian mechanical systems with symmetries has been initiated in \cite{K85,BKMS92} and then further developed in \cite{WK92,BMS97,BLM01a,BLM01b} as-well as, more recently, in \cite{PB19}. The idea of this method consists in modifying the kinetic energy metric of a given mechanical system by means of a Kaluza-Klein construction. The modified metric then yields a new Hamiltonian system and the energy-momentum method of \cite{A66,HMRW85} can be used to find conditions on the Kaluza-Klein construction such that an unstable equilibrium for the uncontrolled system is (nonlinearly) stable for the new system. The method is set up in such a way that the Kaluza-Klein modification can be identified with a feedback control acting on the internal symmetry variables, and the new Hamiltonian system corresponds to the closed loop equations associated to the feedback control. Hence, the control stabilizes a given equilibrium if, and only if, it is stable for the modified Hamiltonian system.  
Stability shall be understood throughout in the nonlinear sense, as in \cite{HMRW85}.

To describe this idea in more detail, consider the example of a satellite with an internal rotor attached to the third principal axis. The configuration space of this system is $P = \SO(3)\times S^1$. Given moments of inertia $I_1>I_2>I_3$, the rotation of the satellite about the second axis is an unstable equilibrium. The kinetic energy of the system is the Hamiltonian function associated to a Kaluza-Klein metric $\mu_0^P$ on the $S^1$-principal bundle $P\to\SO(3)$ which is determined by the following three ingredients: a metric $\mu_0^S$ on $S=\SO(3)$; an inertia tensor $\I_0$ on $\R$ (viewed as the Lie algebra of $S^1$); a connection form $A_0: T\SO(3)\to\R$. Now the control approach of \cite{K85,BMS97,BLM01a,BLM01b} consists of modifying the data $(\mu_0^S,\I_0,A_0)$. This yields a new Kaluza-Klein metric $\mu_C^P$ on $P$, thus a new kinetic energy, and thus a new Hamiltonian system. Moreover, the modification can be identified with a feedback control of the form $q = -C\Pi$ where $q\in\R$ is the angular momentum of the rotor, $\Pi\in\R^3$ is the angular momentum of the satellite in the body representation and $C: \R^3\to\R$ is a linear map. Then it is shown that the closed loop equations associated to $C$ coincide with the Hamiltonian equations with respect to the kinetic energy Hamiltonian of $\mu_C^P$.  Therefore, the energy-momentum method can be used to find a control which stabilizes rotation of the satellite about the middle axis. The details of this example are described in Section~\ref{sec:rigidB}. 

The advantage of the method of controlled Hamiltonians is that it gives an algorithmic and explicit construction of feedback controls which stabilize a given (unstable) equilibrium. It should be noted, however, that this approach only yields stability with respect to perturbations after factoring out the internal symmetries. In the satellite example, this means that stability with respect to perturbations in the rotor variable cannot be concluded. A generalization to show stability in the full phase space has been carried out in \cite{BLM01b}, but this will not be further addressed in this paper.

\subsection{Lie-Poisson formulation of ideal flow of charged particles}\label{sec:introLP}
The theory of \cite{K85,BMS97,BLM01a,BLM01b} applies to mechanical systems where the configuration space is a direct product of two (finite dimensional) Lie groups. We extend this method to treat fluid dynamical systems defined on semi-direct products of infinite dimensional groups. 

To this end, equations \eqref{1e:EM_u_d}-\eqref{1e:cont_d} are reformulated as a Lie-Poisson system~\eqref{e:lp1}. This equivalence is detailed in Appendix~\ref{app:GM}. The Lie-Poisson system will be slightly more general by replacing the electromagnetic gauge symmetry group $S^1$ by an arbitrary finite dimensional compact Lie group $K$. 

More precisely, consider the group $\A$ of volume preserving automorphisms of a trivial principle bundle $P = M\times K \to M$ with base $M\subset\R^n$. This group is a semi-direct product
\[
 \A = \D\circledS\G
\]
where $\D = \textup{Diff}_0(M)$ is the group of volume preserving diffeomorphisms of $M$ and $\G=\F(M,\ko)$ are functions on $M$ with values in the Lie algebra $\ko$ corresponding to the Lie group $K$.  Let $\mu_0^M$ be the induced Euclidean metric on $M$, $\I_0\in\ko^*\otimes\ko^*$ an $\Ad(K)$-invariant symmetric positive definite bilinear form on $\ko$, and $A_0: TM\to\ko$ a connection form on the (trivial) principal bundle $P\to M$. Let $\mu_0^P$ denote the Kaluza-Klein metric on $P$ associated to $(\mu_0^M,\I_0,A_0)$ and $\vol_P$ the associated volume form. 
(See Appendix~\ref{app:GM} for the relevant definitions regarding Kaluza-Klein metrics and Lie-Poisson systems.) 
This gives rise to a right invariant $L^2$ metric $[\mu_0^P]$ on $\A$ and thus to a kinetic energy Hamiltonian $H_0: T^*\A = \A\times\ao^*\to\R$, $(\Phi,\eta)\mapsto\vv<\eta,[\mu_0^P]^{-1}\eta>/2$, where the trivialization $T^*\A = \A\times\ao^*$ follows from right multiplication in $\A$. Due to right invariance, this system can further be viewed as a Lie-Poisson system  
\begin{align*}
\tag{\ref{e:lp1}}
    \dot{\nu} 
    &= -\ad(u)^*\nu - X\diamond q, 
    \quad
    \dot{q}
    = -\rho^u(q) - \ad(X)^*q
    ,\quad 
    \left(\begin{matrix} 
    u \\ 
    X
    \end{matrix}\right)
    =   [\mu_0^P]^{-1}
    \left(\begin{matrix} 
    \nu \\ 
    q 
    \end{matrix}\right)
\end{align*}
on $\ao^* = \du^*\times\gau^*$; here $\du=T_e\D$ and $\gau=T_e\G$, $(u,\nu)\in\du\times\du^*$ and $(X,q)\in\gau\times\gau^*$. 
The first equation is the Euler equation for ideal incompressible fluid flow of charged particles in $M$ under the influence of the external Yang-Mills field $\textup{Curv}^{A_0} = dA_0 + [A_0,A_0]/2$ and the second equation represents conservation of charge. See \cite{GBTV13} for an Euler-Poincar\'e version of this equation and a discussion. 
We also refer to \cite{GHK83} for further background. However, contrary to \cite{GHK83}, we do not include  dynamical equations for the Yang-Mills field. Thus we assume that the motion of the fluid does not influence the field.

\subsection{Feedback control}
To extend the theory of \cite{BMS97,BLM01a,BLM01b} we construct a force $F$ acting on the charge variables $q$ which is of the form
\[
\tag{\ref{e:TCforce}}
 \frac{Dq}{dt} = -F(\nu,q)
\]
and we emphasize the $q$-dependence. This dependence is a new feature compared to the approach of \cite{BMS97,BLM01a,BLM01b}. It is necessary because of the $X\diamond q$ term in the dynamical equation, which, in turn, is due to the semi-direct product structure. 

Using an explicit expression for the force $F$, we obtain a new conserved quantity $p_0$. This allows to identify the corresponding control law as 
\[
\tag{\ref{e:PC3}}
 q = Tp - C\nu 
\]
where $p = \rho^{\phi^{-1}}\Ad(g^{-1})^*p_0$ is defined as the advection of $p_0$, $T:\gau^*\to\gau^*$ is an isomorphism and $C: \du^*\to\gau^*$ is a linear operator. It is shown that, if $T$ and $C$ satisfy the assumptions of Theorem~\ref{thm:1}, then the corresponding closed loop equations coincide with a forced Lie-Poisson system, associated to an explicit force term $f$ acting on the fluid momentum variables $\nu$ and a kinetic energy Hamiltonian $H_C: \ao^*=\du^*\times\gau^*\to\R$, $\eta\mapsto\vv<\eta,[\mu_C^P]^{-1}\eta>/2$. 
The construction is such that 
$[\mu_C^P]$ can be expressed as a Kaluza-Klein inner product on $\ao = \du\times\gau$ which arises as a modification $([\mu_C^M],\I_C,A_C)$ of the data $(\mu_0^M,\I_0,A_0)$. This is the content of Theorem~\ref{thm:1}. 

With the goal of obtaining a Lie-Poisson system on $\du^*$, and accompanying stabilization conditions, we set $p_0=0$. Since $p_0$ is a conserved quantity, this corresponds to a symplectic reduction of $T^*\A$ with respect to the cotangent lifted action of $\G$ at $p_0$, followed by a passage from $T^*\D$ to $\du^*$, which is the Poisson reduction with respect to the remaining $\D$-symmetry. 
In order to have an unforced Lie-Poisson system on $\du^*$, we look for controls such that induced force term $f$ vanishes for $p_0=0$. It turns out that this determines $C$ to be of the form 
\[
\tag{\ref{e:Cgamma}}
 C = \gamma R^{-1} \I_0 A_0 [\mu_0^M]^{-1}: \du^*\to\gau^*
\]
where $\gamma$ is a parameter such that $R: \gau^*\to\gau^*$,  $p\mapsto p - \gamma \I_0 A_0 [\mu_0^M]^{-1}[A_0^*p]$ is invertible; here $[A_0^*p]$ is the class of $A_0^*p = p\circ A_0\in\Om^1(M)$ in $\du^* = \Om^1(M)/d\F(M)$. In fact, as shown in Theorem~\ref{thm:M_CT}, this also fixes $T$ and we obtain 
\[
 \tag{\ref{e:lpa}}
 \dot{\nu} = \ad([\mu_C^M]^{-1}\nu)^*\nu
\]
which is a Lie-Poisson system on $\du^*$ with respect to the kinetic energy Hamiltonian $h_C(\nu) = \vv<\nu,[\mu_C^M]^{-1}\nu>/2$. Note that, once $(\mu_0^M,\I_0,A_0)$ are fixed, the only free parameter in \eqref{e:Cgamma} is $\gamma$. 

Assume now that $\nu_e$ is an unstable equilibrium of the (uncontrolled) Euler equation $\dot{\nu} = \ad([\mu_0^M]^{-1}\nu)^*\nu$. If $\nu_e$ is a (nonlinearly) stable equilibrium of the controlled system \eqref{e:lpa}, then the control $C$ yields stabilization of $\nu_e$ with respect to perturbations in the $\nu$-variables. Since \eqref{e:lpa} is Lie-Poisson, it is in particular Hamiltonian, and the techniques of \cite{A66,HMRW85} can be applied to find conditions on $h_C$ such that stability of $\nu_e$  follows. These conditions translate to explicit conditions on $C$. Therefore, the approach yields a constructive way to design stabilizing feedback controls. 

\begin{remark}[Controllabilty vs.\ stabilizing control]
Finite dimensional control theory makes heavy use of the notion of controllability. This concerns problems such as which states of a system can be reached, when starting from a given initial configuration, by undergoing a series of allowed motions or applying a set of controls. This theory is very well developed in finite dimensions (\cite{AS04,M02}) but there is not much literature on infinite dimensional aspects of controllability or attainability (\cite{AS08} contains a review).     
Consequently, the idea of stabilizing an equilibrium in fluid mechanics is largely independent from the theory of controllability. Furthermore, the control objective in plasma dynamics is confinement, i.e.\ ideally keeping the plasma as close as possible for as long as possible to a fixed initial configuration. Thus one is interested in controlling an external electromagnetic field such that the plasma fluid is held in a stable equilibrium (\cite{C84,Davidson,Lin69,CGG00}).
\end{remark}

\subsection{Feedback control of shear flow}
In Section~\ref{sec:shear} this method is applied to the example of ideal incompressible shear flow with a sinusoidal velocity profile $u_e(x,y) = (\sin(y+\by{\pi}{2}),0)$ in a channel $M = [0,X\pi] \times [0,Y\pi]$ where $Y<1$. This shear flow has an inflection point and is known to be a stable equilibrium of the Euler equation $\dot{\nu} = \ad([\mu_0^M]^{-1}\nu)^*\nu$, where $\mu_0^M$ is the Euclidean metric, if $X$ is sufficiently small. For large $X$, the equilibrium is unstable. Assuming that the fluid consists of charged particles in an external magnetic field $dA_0 = -a_0'(y)\,dx\wedge dy$ and $\I_0=1$, we explicitly construct a control which stabilizes the shear flow for arbitrarily large $X$. 
See figure~\ref{fig:channel}.

\begin{center}
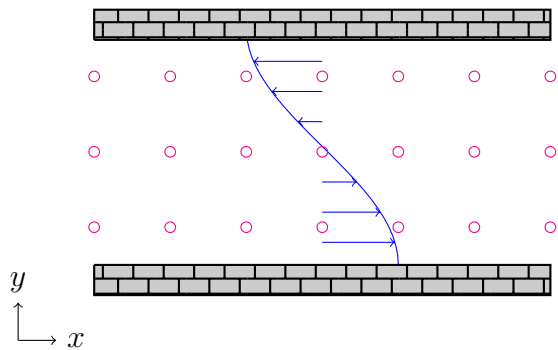
\begin{figure}[h]
\begin{tikzpicture}[domain=0:4]
    \draw[->] (-4,-1) -- (-3.5,-1) node[right] {$x$};
    \draw[->] (-4,-1) -- (-4,-0.5) node[above] {$y$};
    \draw[ domain = 0:(0.95*3.14159),
        color = blue,
        thin]   
        plot (  {sin((\x + 3.14159/2) r)}, \x) ;
    \draw[->, blue] (0,0.3) -- ( 0.95533688121 , 0.3);
    \draw[->, blue] (0,0.7) -- ( 0.76484304202 , 0.7);
    \draw[->, blue] (0,1.1) -- ( 0.45359730387, 1.1);
    \draw[->, blue] (0,1.9) -- ( -0.32328831131 , 1.9);
    \draw[->, blue] (0,2.3) -- ( -0.66627503188, 2.3);
    \draw[->, blue] (0,2.7) -- ( -0.90407157497 , 2.7);
        \foreach \x in {-3,-2,...,3} {
        \foreach \y in {0.5,1.5,...,2.5} {
            \draw[magenta] (\x,\y) circle (0.07cm);
        }
    }
    \def\wallA{ (-3,0) -- (3,0) -- (3,-0.4) --
        (-3,-0.4) -- (-3,0) }
        \draw[thick,fill=gray!40] \wallA;
        \pattern[pattern=bricks,pattern color=black]
        \wallA;
    \def\wallB{ (-3, 0.95*3.14159) -- (3,0.95*3.14159) --         (3,0.95*3.14159 + 0.4) -- (-3,0.95*3.14159 + 0.4) 
         -- (-3, 0.95*3.14159) }
        \draw[thick,fill=gray!40] \wallB;
        \pattern[pattern=bricks,pattern color=black]
        \wallB;
\end{tikzpicture}
\caption{The shear flow velocity profile of $u_e = u_e(y)\,e_x$ is shown in blue. The circles indicate the magnetic field which points out of the paper.
The channel width in this picture is $0.95\cdot\pi$.}
\label{fig:channel}
\end{figure}
\end{center}

Since the only free parameter in \eqref{e:Cgamma} is $\gamma$, this amounts to finding a vector  potential $A_0 = a_0(y)\,dx$ such that $C$ acts stabilizing for $\pm\gamma$ sufficiently large. An explicit formula for $A_0$, with $\gamma=1$, is given in Theorem~\ref{thm:shear_formal_stab}: $a_0(y) = b(\om_e(y))$ where $\om_e(y) = -\cos(y+\by{\pi}{2})$ is the vorticity function associated to $u_e$ and $b$ is a linear map defined in terms of the channel length $X\pi$ and width $Y\pi$.

\subsection{Structure of the paper}
Section~\ref{sec:rigidB} is a detailed exposition of the satellite example of \cite{BMS97,MarsdenPCL} mentioned above. The notation is chosen so that the comparison with Section~\ref{sec:FB_IS} is straightforward. Moreover, we emphasize a point that is mentioned in \cite{BMS97,MarsdenPCL}, but not explained very prominently: the identification of the closed loop equations with the Hamiltonian equations associated to the modified Kaluza-Klein metric involves also a change in the momentum variable. See Remarks~\ref{rem1} and \ref{rem3}.
Because of the semi-direct product structure, it is necessary to systematize this change in the momentum variables by the use of the isomorphism $T: \gau^*\to\gau^*$ in Section~\ref{sec:FB_IS}.

Section~\ref{sec:FB_IS} contains the extension of the theory of controlled Hamiltonians to the case of incompressible ideal fluids under the influence of an external Yang-Mills field. The main results are Theorem~\ref{thm:1}, which provides the link between the closed loop equations and the (forced) Lie-Poisson system, and Theorem~\ref{thm:M_CT}, which shows that there is a control $q=T p - C \nu$ such that the reduction to $\du^*$ yields a true (unforced) Lie-Poisson system~\eqref{e:lpa}; moreover $C$ and $T$ are given explicitly. 

Section~\ref{sec:stab} collects some results concerning the stability of equilibria from \cite{A66,AK98,HMRW85} and provides the context so that these results can be applied to the system~\eqref{e:lpa}. 

Section~\ref{sec:RBagain} treats again the satellite with a rotor example, but this time from the combined point of view of Sections~\ref{sec:FB_IS} and \ref{sec:stab}. 

Section~\ref{sec:stab_2D} provides further background on the stability of equilibria for two-dimensional flows. Thus we assume that $M\subset\R^2$ and adapt the two-dimensional results of \cite{A66,AK98} to the case of Section~\ref{sec:FB_IS}. 

Section~\ref{sec:shear} contains the shear flow example. We consider incompressible ideal flow of charged particles in an external magnetic field. In Theorem~\ref{thm:shear_formal_stab} it is shown how the approach of Section~\ref{sec:FB_IS}, together with the background from Section~\ref{sec:stab_2D}, yields a control on the charge such that the flow is stabilized with respect to perturbations in the fluid momentum variables.

\subsection*{Acknowledgements}
I am very grateful to Darryl Holm for drawing my attention to \cite{BMS97,MarsdenPCL} and all the patient explanations, and to Florian Gach for all the helpful discussions.

\section{Feedback control of the rigid body with a rotor}\label{sec:rigidB}
This section is a detailed account of the satellite with a rotor example in \cite{BMS97,MarsdenPCL}. Its purpose is twofold: Firstly, to compare the construction of Section~\ref{sec:FB_IS}. Secondly, to explain that the control mechanism of \cite{BMS97,MarsdenPCL} involves not only a Kaluza-Klein construction, but also a change in the momentum variable (Remarks~\ref{rem1} and \ref{rem3}).

\subsection{The free system}\label{sec:free_sat}
Let $S=\textup{SO}(3)$, $K=S^1$ and $P=S\times K$ the configuration space of the rigid body with one rotor about the third principal axis. 
Let $I_1 > I_2 > I_3$ be the rigid body moments of inertia and $i_1 = i_2 > i_3$ those of the rotor. We use left multiplication in the direct product group $P$ to write the tangent bundle $TP \cong P\times\mathfrak{so}(3)\times\R \cong P\times\R^4$ in body coordinates $(g,\alpha,\Omega,x)$. This means that $g$ gives the orientation of the body, $\Omega$ the body angular velocity, $\alpha$ the relative angle of the rotor and $x$ the rotor angular velocity. The metric tensor is
\begin{equation}
{\mu^P_0}
 =
    \left(\begin{matrix}
    \lam_1 & & & \\
     & \lam_2 & & \\
     & & \lam_3 & i_3 \\
     & & i_3 & i_3
    \end{matrix}\right)
    \textup{ with inverse }
    (\mu^P_0)^{-1}
 =
    \left(\begin{matrix}
    \lam_1^{-1} & & & \\
     & \lam_2^{-1} & & \\
     & & I_3^{-1} & -I_3^{-1} \\
     & & -I_3^{-1} & i_3^{-1}+I_3^{-1}
    \end{matrix}\right)
\end{equation}
where $\lam_j = I_j+i_j$.
The  equations of motion are determined by the free Hamiltonian system 
$(T^*P, \Om^{T^*P}, H_0)$ where $\Om^{T^*P}$ is the canonical symplectic form, and
\begin{equation}
    H_0: 
    \left(\begin{matrix}\Pi \\ 
     q  \end{matrix}\right)
     \mapsto 
     \frac{1}{2}\left<
     (\mu^P_0)^{-1}  \left(\begin{matrix}\Pi \\ 
     q  \end{matrix}\right),
         \left(\begin{matrix}\Pi \\ 
     q  \end{matrix}\right) \right>
\end{equation}
where $\Pi$ and $q$ are the body and rotor angular momenta, respectively.
Since $P$ is a Lie group, the equations of motion are given by 
\begin{equation}
\label{e:lpRB}
     \left(\begin{matrix}
     \dot{\Pi} \\ 
     \dot{q}  
     \end{matrix}\right)
     =
    \ad
    \left(\begin{matrix}
    \Omega \\ 
     x  \end{matrix}\right)^*
    \left(\begin{matrix}\Pi \\ 
     q  \end{matrix}\right),
     \qquad
       \left(\begin{matrix}
    \Omega \\ 
     x  \end{matrix}\right)
     = (\mu_0^P)^{-1}   
     \left(\begin{matrix}\Pi \\ 
     q  \end{matrix}\right)
     = 
    \left(\begin{matrix}
    g^{-1}\dot{g} \\ 
    \dot{\alpha}   
    \end{matrix}\right).
\end{equation}
Explicitly, since $\ad(a)^*b = -\ad(a)b = -[a,b]$ and $(\mathfrak{so}(3),[.,.]) = (\R^3,\times)$, this is
\begin{align}
    \dot{\Pi} 
    &= -\Om\times\Pi 
    = 
    \left(\begin{matrix}
    (-\lam_2^{-1}+I_3^{-1})\Pi_2\Pi_3 - I_3^{-1} q\Pi_2 \\
    (\lam_1^{-1}-I_3^{-1})\Pi_1\Pi_3 + I_3^{-1} q\Pi_1 \\
    (-\lam_1^{-1} + \lam_2^{-1})\Pi_1\Pi_2 
    \end{matrix}\right)
    \label{e:eom1}\\
    \dot{q} &= 0.   \label{e:eom2}
\end{align}
Consider the action by $K=S^1$ on $P=\textup{SO}(3)\times K$ given by the action on the second factor. This action leaves the metric $\mu^P_0(.,.) = \vv<{\mu^P_0}.,.>$ invariant and $\pi: P\to S$ is a principal bundle.
There is a natural connection on this bundle given by 
the splitting
\begin{equation}
    TP = \textup{Hor}_0 \oplus \textup{Ver},
\end{equation}
where $\textup{Ver} = \textup{ker}(T\pi)$ and $\textup{Hor}_0$ is the orthogonal complement with respect to $\mu^P_0$. This is the so-called mechanical connection. Let  
\begin{equation}
    A_0: TS\to \ko = \R,\;
    \Om\mapsto \Om_3
\end{equation}
denote the associated (local) connection form. The associated (local) curvature form is denoted by $K_0 = dA_0$. Note that this is a two-from on $S$ and satisfies
\begin{equation}\label{e:K0}
    i(\Omega)K_0 
    =
    \left(\begin{matrix}
    \Om_2\\
    -\Om_1\\
    0
    \end{matrix}\right)
\end{equation}
The system $(T^*P,\Om^{T^*P},H_0)$ is invariant under the cotangent lifted action of $K$. 
Let 
\[
\textup{Hor}^* := \textup{Ann}(\textup{Ver})
\textup{ and }
\textup{Ver}_0^* := \textup{Ann}(\textup{Hor}_0)
\]
Note that $\textup{Hor}^*$ is canonically defined, while $\textup{Ver}_0^*$ depends on the choice of connection. 
Consider the connection dependent isomorphism
\begin{equation}
    \Psi_0:
    T^*P \cong \textup{Hor}^*
    \oplus\textup{Ver}_0^*
    \cong P\times_{S}T^*S \times \ko^* 
\end{equation}
given by 
\begin{equation}
     \left(\begin{matrix}
     \Pi \\ 
     q  
     \end{matrix}\right)
    \mapsto
      \left(\begin{matrix}
      \Pi_1 \\
      \Pi_2 \\
      \Pi_3 -  q\\ 
      0  
      \end{matrix}\right)
    \oplus
    \left(\begin{matrix}
      0 \\
      0 \\
      q \\ 
      q  
    \end{matrix}\right)
      \mapsto
      (\Phi, p_0)
\end{equation}
where $\Phi_1=\Pi_1$, $\Phi_2=\Pi_2$, $\Phi_3 = \Pi_3-q$ and $p_0 = q$. In these coordinates the equations of motion \eqref{e:eom1}, \eqref{e:eom2} become
\begin{align}
    \dot{\Phi}  
    &= 
    \left(\begin{matrix}
    (-\lam_2^{-1}+I_3^{-1})\Phi_2\Phi_3 - \lam_2^{-1}p_0\Phi_2 \\
    (\lam_1^{-1}-I_3^{-1})\Phi_1\Phi_3 + \lam_1^{-1}p_0\Phi_1 \\
    (-\lam_1^{-1} + \lam_2^{-1})\Phi_1\Phi_2 
    \end{matrix}\right)
    =
    \left(\begin{matrix}
    (-\lam_2^{-1}+I_3^{-1})\Phi_2\Phi_3  \\
    (\lam_1^{-1}-I_3^{-1})\Phi_1\Phi_3 \\
    (-\lam_1^{-1} + \lam_2^{-1})\Phi_1\Phi_2 
    \end{matrix}\right)
    -
    p_0 \left(\begin{matrix}
    \lam_2^{-1}\Phi_2 \\
    -\lam_1^{-1}\Phi_1 \\
    0 
    \end{matrix}\right)
    \notag \\
    &=
    \ad((\mu_0^S)^{-1}\Phi)^*.\Phi 
    - \vv<p_0, i((\mu_0^S)^{-1}\Phi)K_0>
    \label{e:eom3}\\
    \dot{p}_0 &= 0 .  \label{e:eom4}
\end{align}
where $\mu_0^S$ is defined as follows.
Let $\textup{hl}^0: P\times TS\to\textup{Hor}_0$ denote the horizontal lift map associated to $A_0$, that is $\textup{hl}^0: \Om\mapsto(\Om,-\Om_3)$. Define the induced metric $\mu_0^S$ on $S$ by 
\begin{equation}
\label{e:muh}
    \mu_0^S(\Om,\tilde{\Om}) 
    = \mu^P_0( \textup{hl}^0_{\Om} , \textup{hl}^0_{\tilde{\Om}} ) 
    = \vv<\diag{\lam_1,\lam_2,I_3}\,\Om, \tilde{\Om}> 
\end{equation}
such that $\pi: (P,\mu^P_0)\to (S,\mu^S_0)$ is a Riemannian submersion.  
Consider the associated Hamiltonian function
\begin{equation}
    h_0: 
    T^*S\times\ko^* \to \R,
    (\Phi,p_0) \mapsto 
    \by{1}{2}\vv<({\mu}_0^S)^{-1}\Phi,\Phi> + \by{1}{2}i_3^{-1} p_0^2. 
\end{equation}
Then equations \eqref{e:eom3} and \eqref{e:eom4} are the Hamiltonian equations associated to $h_0$ and the following direct product Poisson structure: on $\ko^* = \R$ we consider the trivial Poisson structure (whose symplectic orbits are points); on $T^*S$ we consider the magnetic symplectic form
\begin{equation}
 \Om^0 = \Om^{T^*S} - \vv<p_0, K_0>     
\end{equation}
Indeed, this follows immediately from ${\mu}_0^S = \textup{diag}(\lam_1,\lam_2,I_3)$ and equation~\eqref{e:K0}. 

This construction can be summed up by saying that $(T^*S\times\set{p_0},\Om^0,h_0)$ is the Hamiltonian reduction of $(T^*P,\Om^{T^*P},H_0)$ at $p_0$ with respect to the $K$-action.

\subsection{Feedback control via magnetic reduction}
Consider the controlled equations 
\begin{align}
    \dot{\Pi} 
    &= -\Om\times\Pi 
    = 
    \left(\begin{matrix}
    (-\lam_2^{-1}+I_3^{-1})\Pi_2\Pi_3 - I_3^{-1} q\Pi_2 \\
    (\lam_1^{-1}-I_3^{-1})\Pi_1\Pi_3 + I_3^{-1} q\Pi_1 \\
    (-\lam_1^{-1} + \lam_2^{-1})\Pi_1\Pi_2 
    \end{matrix}\right)
    \label{e:c-eom1}\\
    \dot{q} &= \mathcal{U}   \label{e:c-eom2}
\end{align}
where $\mathcal{U}$ is the control. Following \cite{BMS97,MarsdenPCL} we show how certain feedback controls $\mathcal{U}$ can be obtained from a Kaluza-Klein construction. 
Let $k$ be a parameter and $\var_k$ a number such that $\var_0=1$. Then we define a new (local) connection form
\begin{equation}
    A_k: TS \to \ko,\;
    \Om\mapsto \var_k \Om_3
\end{equation}
giving rise to a new horizontal bundle $\textup{Hor}_k = \set{(\Om,x)\in TP: x = -\var_k\Om_3}$. We use the splitting $TP = \textup{Hor}_k\oplus\textup{Ver}$ to define a new metric $\mu^P_k$ on $P$:
\begin{itemize}
    \item 
    Let $\mu_k^S$ be a metric on $S$. We require that $T\pi: \textup{Hor}_k\to TS$ is an isometry. This defines $\mu_k^P$ on horizontal vectors.
    \item
    Let $\mathbb{I}_k$ be an inner product on $\ko$. We require that $\mu_k^P(\zeta_x,\zeta_y) = \mathbb{I}_k(x,y)$ for all $x,y\in\ko$, where $\zeta:\ko\to\X(P)$ is the fundamental vector field map associated to the $K$-action.
    \item
    $\textup{Hor}_k$ and $\textup{Ver}$ shall be orthogonal with respect to $\mu_k^P$.
\end{itemize}
The orthogonality condition is important since we want $A_k$ to be a mechanical connection in order to apply the (magnetic) Hamiltonian reduction procedure.
We call $\mu_k^P = \mu^{KK}(\mu_k^S,\I_k,A_k)$ the Kaluza-Klein metric associated to $(\mu_k^S,\I_k,A_k)$.
Further, the metric $\mu_k^S$ should be left-invariant and we assume that the metric tensor 
\begin{equation}\label{e:muS}
\mu_k^S
 =
    \left(\begin{matrix}
    \tilde{\lam}_1 & &  \\
     & \tilde{\lam}_2 & \\
     & & \tilde{I}_3  \\
    \end{matrix}\right)
\end{equation}
is of diagonal form.
Since $\ko=\R$ the vertical part of the metric is determined by a number $\mathbb{I}_k>0$.
Let $H_k$ denote the natural Hamiltonian with respect to $\mu_k^P$. Since $(T^*P,\Om^{T^*P},H_k)$ is still invariant under the $K$-action, with the same momentum map $J:T^*P\to\ko^*$, we can carry out Hamiltonian reduction at a level $\tilde{p}_k\in\ko^*$. This yields, exactly as above, the equations of motion
\begin{align}
    \dot{\Phi}  
    &=
    \left(\begin{matrix}
    (-\tilde{\lam}_2^{-1}+\tilde{I}_3^{-1})\Phi_2\Phi_3  \\
    (\tilde{\lam}_1^{-1}-\tilde{I}_3^{-1})\Phi_1\Phi_3 \\
    (-\tilde{\lam}_1^{-1} + \tilde{\lam}_2^{-1})\Phi_1\Phi_2 
    \end{matrix}\right)
    -
    \var_k\tilde{p}_k \left(\begin{matrix}
    \tilde{\lam}_2^{-1}\Phi_2 \\
    -\tilde{\lam}_1^{-1}\Phi_1 \\
    0 
    \end{matrix}\right)
    \label{e:c-eom3}\\
    \dot{\tilde{p}}_k &= 0 .  \label{e:c-eom4}
\end{align}
Let us now assume $\tilde{\lam}_1 = \lam_1$, $\tilde{\lam}_2=\lam_2$, $\tilde{I}_3 = (1-k)^{-1}I_3$ and $\tilde{p}_k = (1-k)^{-1}\var_k^{-1}p_k$. 
With the assignment $\Pi_1=\Phi_1$, $\Pi_2=\Phi_2$ and $\Pi_3=\Phi_3+\var_k\tilde{p}_k$, which  corresponds to the $A_k$-dependent isomorphism $\textup{Hor}^*\oplus\textup{Ver}^*_k\to T^*P$, 
equations~\eqref{e:c-eom3} can be rearranged to give 
\begin{align}
    \dot{\Pi}  
    &=
    \left(\begin{matrix}
    -\lam_2^{-1}\Pi_2\Pi_3 + I_3^{-1}\Pi_2((1-k)\Pi_3 - p_k) \\
    \lam_1^{-1}\Pi_1\Pi_3 - I_3^{-1}\Pi_1((1-k)\Pi_3 - p_k) \\
    (-\lam_1^{-1} + \lam_2^{-1})\Pi_1\Pi_2 
    \end{matrix}\right)
    \label{e:c-eom5}
\end{align}
which are the closed loop equations corresponding to \eqref{e:c-eom1}, \eqref{e:c-eom2} with respect to the control  
$\mathcal{U} = k(-\lam_1^{-1} + \lam_2^{-1})\Pi_1\Pi_2 = k\dot{\Pi}_3$. 
That is 
\begin{equation}
\label{e:controlPk}
    q = p_k+k\Pi_3 
\end{equation}
for a constant $p_k$.
The control law yields a new conserved quantity $p_k=q-k\Pi_3$ and, because of \eqref{e:diag1}, this implies that $\var_k$ is given by \eqref{e:vark}. 

\begin{remark}\label{rem1}
Note that we had to change the momentum value from $p_k = q-k\Pi_3$ to $\tilde{p}_k$ such that $\var_k\tilde{p}_k = (1-k)^{-1}p_k$. 
This means that the controlled equations \eqref{e:c-eom5} are \emph{not} obtained by replacing \eqref{e:eom3} with
$\dot{\Phi} 
 = \ad((\mu_k^S)^{-1}\Phi)^*\Phi - \vv<p_k, i((\mu_k^S)^{-1}\Phi)K_k>$, 
where $K_k = \var_k K_0$ is the curvature of $A_k = \var_k A_0$. 
This is consistent with \cite{BMS97,MarsdenPCL} and the factor of $1-k$ is mentioned in the sentence immediately after \cite[Equ.~(3.7)]{BMS97}.
\end{remark}

\begin{remark}[Physical significance of the control]\label{rem:phys}
The control \eqref{e:c-eom2} means that the angular velocity of the rotor which spins about the third principal axis, as introduced at the beginning of Section~\ref{sec:free_sat}, is controlled. This control is of feedback type and depends on observations of the carrier rigid body's angular momentum. The physical intuition behind this technique (as-well as applications in the attitude control of satellites) are described in detail in \cite{K85}.
\end{remark}

\subsection{Controlling the conserved quantity: Lie-Poisson approach}
Fix a $k$-dependent linear map $C_k: \so(3)^*=\R^3\to\ko^* = \R$ and consider
\begin{equation}
 \label{e:C_k}
 J_k: T^*P\to\ko^*,\;
 (\Pi,q)\mapsto J(\Pi,q) + C_k(\Pi) = q+C_k(\Pi).
\end{equation}
In the above example, we have 
\begin{equation}
\label{e:C_k1}
    C_k(\Pi) = -k\Pi_3.
\end{equation} 
The map $C_k$ determines the feedback control law by requiring $p_k = J_k(\Pi,q)$ to be constant. Thus, $C_k$ provides a new conserved quantity, it is however not a momentum map. The momentum map remains unchanged and is $J: (\Pi,q)\mapsto q$.  With $C_k(\Pi)=-k\Pi_3$ and notation as above, the conservation of $J_k$ can be linked to a conservation law associated to a Kaluza-Klein metric $\mu_k^P = \mu^{KK}(\mu_k^S,\I_k,A_k=\var_k A_0)$ if, and only if, the following diagram commutes:
\begin{equation}\label{e:diag1}
\xymatrix
   { 
   (\Om,x)
        \ar@{=}[r]
        \ar@{|->}[d]^{{\mu}_0^P} 
    & (\Om,x)
        \ar@{|->}[d]^{{\mu}_k^P}
   \\
   (\lam_1\Om_1,\lam_2\Om_2,\lam_3\Om_3+i_3x,i_3\Om_3+i_3x)
        \ar@{|->}[d]^{J_k}
   &
   (\tilde{\lam}_1\Om_1,\tilde{\lam}_2\Om_2,
   (\tilde{I}_3+\var_k^2\mathbb{I}_k)\Om_3 + \var_k\mathbb{I}_k x 
   , \var_k\mathbb{I}_k\Om_3 + \mathbb{I}_k x)
        \ar@{|->}[d]^{J}
   \\
   i_3(1-k)x+(i_3-k\lam_3)\Om_3
        \ar@{==}[r]
   & \mathbb{I}_k(\var_k\Om_3 + x)
   }
\end{equation}
This holds if $\mathbb{I}_k = i_3(1-k)$ and  
\begin{equation}
    \label{e:vark}
    \var_k = \mathbb{I}_k^{-1}(i_3-k\lam_3). 
\end{equation}
In terms of angular velocities, the feedback control law $J_k = const.$ is thus 
\begin{equation}
    \dot{x} = -\frac{i_3-k\lam_3}{i_3(1-k)}\dot{\Om}_3.
\end{equation}

Now, the map \eqref{e:C_k1} yields a conserved quantity $p_k$, such that 
$\dot{q} + C_k(\dot{\Pi}) = \dot{p}_k = 0$. 
The closed loop equations associated to the corresponding control 
$q = p_k - C_k(\Pi)$ and \eqref{e:lpRB} are
\begin{equation}
\label{e:loop1}
    \dot{\Pi} 
    = \ad\Big((\pr_1\circ({\mu}_0^P)^{-1})(\Pi,q)\Big)^*\Pi
    = \ad\Big((\pr_1\circ({\mu}_0^P)^{-1})(\Pi,p_k-C_k(\Pi))\Big)^*\Pi
\end{equation}
Thus we have to find a constant $\tilde{p}_k$ and $\mu_k^P = \mu^{KK}(\mu_k^S,\I_k,A_k)$ such that 
\begin{equation}
    \dot{\Pi} 
    = \ad\Big((\pr_1\circ({\mu}_k^P)^{-1})(\Pi,\tilde{p}_k)\Big)^*\Pi
    = \eqref{e:loop1}
\end{equation}
Because $\mathbb{I}_k$ and $A_k$ are already determined by \eqref{e:diag1} we have to find a suitable $\mu^S_k$. 
Consider again the connection dependent isomorphism
\begin{equation}
    \Psi_k: TP\to \textup{Hor}_k\oplus\textup{Ver},\,
    (\Om,x)\mapsto ((\Om,-\var_k\Om_3); (0,x+\var_k\Om_3))
\end{equation}
and the dual isomorphism 
\begin{equation}
    \Psi_k^*: T^*P\to \textup{Hor}^*\oplus\textup{Ver}_k^*,\;
    (\Pi,q)\mapsto((\Pi_1,\Pi_2,\Pi_3-\var_k q,0); (0,0,\var_k q,q))
\end{equation}
Notice that 
\begin{equation}\label{e:metricF_RB}
    ({\mu}_k^P)^{-1}
    = \Psi_k^{-1}\circ
    \left(\begin{matrix}
    ({\mu}_k^S)^{-1} & 0\\
    0 & \mathbb{I}_k^{-1}
    \end{matrix}\right)
    \circ
    \Psi_k^*. 
\end{equation}
Therefore, the defining equation for $\mu_k^S$ and $\tilde{p}_k$ is 
\begin{align*}
    (\pr_1\circ({\mu}_0^P)^{-1})(\Pi,p_k-C_k(\Pi))
    &= (\lam_1^{-1}\Pi_1,\lam_2^{-1}\Pi_2,I_3^{-1}((1-k)\Pi_3 - p_k))\\
    &= (\pr_1\circ\Psi_k^{-1})\Big((({\mu}_k^S)^{-1}\oplus\mathbb{I}_k^{-1})((\Pi_1,\Pi_2,\Pi_3-\tilde{p}_k); p_k)\Big) \\
    &= ({\mu}_k^S)^{-1} (\Pi_1,\Pi_2,\Pi_3-\tilde{p}_k)
\end{align*}
which yields equation~\eqref{e:muS} with $\tilde{\lam}_1 = \lam_1$, $\tilde{\lam}_2=\lam_2$, $\tilde{I}_3 = (1-k)^{-1}I_3$, that is
\begin{equation}\label{e:muS2}
\mu_k^S
 =
    \left(\begin{matrix}
    \lam_1 & &  \\
     & \lam_2 & \\
     & & (1-k)^{-1}I_3  \\
    \end{matrix}\right)
\end{equation}
and 
\begin{equation}
    \label{e:pk_tilde}
    \tilde{p}_k 
    = (1-k)^{-1}\var_k^{-1}p_k
    = \frac{i_3}{i_3-k\lam_3}p_k. 
\end{equation}

\begin{remark}
The law \eqref{e:C_k} means that $\dd{t}{}J = -C_k(\dot{\Pi}) = k(\lam_1^{-1}-\lam_2^{-1})\Pi_1\Pi_2$. Thus the control is physically given by a force acting in the symmetry direction. Compare with \cite[Equation~(1.10)]{BMS97}.
\end{remark}

It follows that the closed loop equation~\eqref{e:loop1} for $\Pi$ and with control constant $p_k$ is equivalent to the system of Hamiltonian (Lie-Poisson) equations 
\begin{equation}
    \label{e:HamLoop1}
    \frac{\del}{\del t}\left(\begin{matrix}
    \Pi \\
    \tilde{p}_k
    \end{matrix}\right)
    =
    \ad\left((\mu^P_k)^{-1}\left(\begin{matrix}
    \Pi \\
    \tilde{p}_k
    \end{matrix}\right)\right)^*
    \left(\begin{matrix}
    \Pi \\
    \tilde{p}_k
    \end{matrix}\right)
\end{equation}
on the direct product dual Lie algebra $\so(3)^*\times\R$. 

\begin{remark}\label{rem3}
Note that not only the Kaluza-Klein data $(\mu_0^S,\I_0,A_0)$ are changed but also the value of the conserved quantity. Compare with Remark~\ref{rem1}. 
\end{remark}

\subsection{Stabilization about the middle axis}\label{sec:RBstab}
Equation~\eqref{e:HamLoop1} is the Lie-Poisson version of Hamiltonian equations (with magnetic term) \eqref{e:c-eom3} and \eqref{e:c-eom4}. This means that the energy-momentum method can be used to analyze stability of equilibria. 

If $\tilde{p}_k = (1-k)^{-1}\var_k^{-1}p_k = 0$, the first line of equation~\eqref{e:HamLoop1} is the Lie-Poisson equation of motion of the rigid body with respect to the moment of inertia tensor~\eqref{e:muS2}. Thus $(0,M,0)^{\top}$ is an equilibrium solution, and this equilibrium is stable if $(1-k)^{-1}I_3 = \tilde{I}_3 > \tilde{\lam}_2 = \lam_2$:

\begin{proposition}[Prop.~3.1 in \cite{BMS97}]\label{prop:RBstab}
Let $1 > k > 1-I_3/\lam_2$ (i.e.\  $(1-k)^{-1}I_3 > \lam_2$) and $q=k\Pi_3$ (i.e.\ $\tilde{p}_k=0$), then the control $C_k(\Pi) = -k\Pi_3$ stabilizes the motion such that $(0,M,0)^{\top}$ becomes a (nonlinearly) stable equilibrium. 
\end{proposition}

The rotor about the third axis increases the third moment of inertia (in the abstract system~\eqref{e:HamLoop1}) and if we turn it fast enough, the moment of inertia about the third axis becomes larger than that about the second axis, yielding new stability properties.

\section{Feedback control of fluids with internal symmetries}\label{sec:FB_IS}

\subsection{Semi direct product structure}\label{sec:SD_struc}
Let $K$ be a compact finite dimensional Lie group and $M$ a compact domain, possibly with boundary,  in $\R^n$. Consider the trivial principal bundle $P := M\times K \to M$ where the principal bundle action is given by right multiplication in the group. Let $\mu_0^M = \vv<.,.>$ denote the induced Euclidean metric on $M\subset\R^n$ and $\I_0$ a symmetric positive definite bilinear form on $\ko$ which is $\Ad(k)$-invariant. We fix a connection form $A_0: TM\to\ko$. The horizontal space $\textup{Hor}_0\subset TP = M\times K\times\R^n\times\ko$ is thus $\textup{Hor}_0 = \set{(u,k,u_x,X): X+A_0(x)u_x = 0}$.
Denote the  Kaluza-Klein metric on $M\times K$ associated to $(\mu_0^M,\I_0,A_0)$ 
by
\begin{equation}
    \label{e:mu0}
    \mu_0^P = \mu^{KK}(\mu_0^M,\I_0,A_0).
\end{equation}
For $(u_x,X), (v_x,Y)\in T_{(x,e)}P = \R^n\times\ko$ this means
\[
 \mu_0^P\Big((u_x,X),(v_x,Y)\Big) 
 = \Big\langle u_x,v_x \Big\rangle + \I_0\Big(X + A_0(x)(u_x),Y + A_0(x)(v_x)\Big).
\]
We refer to a triple, such as   $(\mu_0^M,\I_0,A_0)$, consisting of a Riemannian metric, a symmetric $\Ad(K)$-invariant positive definite bilinear form and a connection as a set of Kaluza-Klein data on the principal bundle $P\to M$.
See Appendix~\ref{app:GM}.
Let $\textup{vol}_P$, $\textup{vol}_M = dx$ denote the volume forms on $P$, $M$ with respect to $\mu_0^P$, $\mu_0^M$ respectively.

\begin{remark}
In the following, the data $(\mu_0^M,\I_0,A_0)$ will be changed. However the volume forms will be kept fixed throughout. Thus the divergence of $v\in\X(M)$ will be with respect to $\textup{vol}_M$, while the divergence of $v\in\X(P)$ will be with respect to $\textup{vol}_P$.
\end{remark}

For $k\in K$, let $r^k: P\to P$, $(x,g)\mapsto(x,gk)$ denote the principal right action. Consider the volume preserving automorphisms  
\[
 \A := \textup{Aut}_0(P)
 := \set{\Phi\in\textup{Diff}(P):
 \Phi\circ r^k = r^k\circ\Phi \; \forall k\in K \And \Phi^*\textup{vol}_P=\textup{vol}_P} 
\]
which, as explained in \cite{GBTV13}, can be identified as 
\[
 \A = \textup{Diff}_0(M)\circledS\F(M,K) = \D\circledS\G
 \]
where $\D := \textup{Diff}_0(M)$ is the set of $\textup{vol}_M$-preserving diffeomorphisms and
$\G := \F(M,K)$ denotes functions from $M$ to $K$ (of a fixed differentiability class which we do not specify). The semi-direct product structure is given in \eqref{e:s-d}. 
Composition from the right gives rise to a right  representation 
\[
 \rho: \D \to \textup{Aut}(\G),\quad
 \phi\mapsto \rho^{\phi}
\]
where $\rho^{\phi}(g) = g\circ\phi$. 
 
Consider the action by point-wise right multiplication $R^g: \G\to\G$, $h\mapsto hg$.
The  induced right action on $\A$ is again denoted by $R$:
\begin{equation}\label{e:s-d}
 R^{(\psi,g)}(\phi,h) = (\phi\cdot\psi,R^g(\rho^{\psi}(h))). 
\end{equation}
In particular, $\A\to\A/\G = \D$ is a right principal $\G$ bundle. 
We also consider the right trivializations
\begin{equation}
    \label{e:triv}
    T\A \cong \A\times\ao = \D\times\du\times\G\times\gau
    \textup{ and }
    T^*\A \cong \A\times\ao^* = \D\times\du^*\times\G\times\gau^* 
\end{equation}
using the right multiplication in $\A$, where $\ao = T_e\A$, 
$\du = T_e\D = \X_0(M)$ are divergence free vector fields tangent to the boundary  
and $\gau=T_e\G = \F(M,\ko)$. Here, $\ao^*$, $\du^*$ and $\gau^*$ denote the smooth part of the dual.

We have $\du^* = \Om^1(M)/d\mathcal{F}(M)$ and $\gau^* = \F(M,\ko^*)$. The pairings are given by 
\begin{align*}
    &\du^*\times\du\to\R,\quad ([\Pi],u)\mapsto\int_M\vv<\Pi_x,u_x>\,dx\\
    &\gau^*\times\gau\to\R,\quad (q,X)\mapsto\int_M\vv<q_x,X_x>\,dx
\end{align*}
where $[\Pi]$ is the class of $\Pi\in\Om^1(M)$.
By definition, the smooth duals are the isomorphic images of the maps
\[
 [\mu_0^M]: \du\to\du^*,\quad
  u \mapsto [\mu_0^M(u)]
\]
and 
$
 \I: \gau\to\gau^*
$,
where $[\mu_0^M(u)]$ is the class of $\mu_0^M(u)\in\Om^1(M)$ in $\Om^1(M)/d\mathcal{F}(M) = \du^*$. The inverse is 
\[
 [\mu_0^M]^{-1}: [\Pi]\mapsto \mathcal{P}(\mu_0^{-1}(\Pi))
\]
where $\Pi$ is a representative of $[\Pi]$ and $\mathcal{P}$ is the Helmholtz-Hodge-Leray projection.
For a vector field $u\in\X(M)$ the Helmholtz-Hodge-Leray projection is divergence free, tangent to the boundary, and given by $\mathcal{P}(u) = u - \nabla g$ where $g$ is determined by $\Delta g = \textup{div}\,u$ with Neumann boundary conditions.


The representation $\rho$ gives rise to an infinitesimal representations
$\rho^{\phi}X$ and 
$\rho^u(X) = dX.u = L_u X = \nabla_u X$ with $\phi\in\D$, $u\in\du$ and $X\in\gau$. The corresponding coadjoint representations are given by $\rho^{\phi}(q) = (\rho^{\phi^{-1}})^*(q)$ and $\rho^u(q) = (\rho^{-u})^*(q)$ with $q\in\gau^*$. 

We define the bracket $[.,.]$ on $\du$ (and similarly for $\ao$) to be the negative of the usual Lie bracket:
$
 [u,v] := -\nabla_{u}v + \nabla_{v}u
$
where $\nabla_{u}v=\vv<u,\nabla>v$ and $u,v\in\du$. This choice of sign is compatible with \cite{A66,AK98}. 
Further, we define the operator $\ad(u).v = [u,v]$. Its dual is $\ad(u)^*[\Pi] = [\Pi\circ\ad(u)]$ for $[\Pi]\in\du^*$.

\begin{remark}
Contrary to Section~\ref{sec:rigidB}, we now perform all calculations in the right trivialization. 
\end{remark}

\subsection{Lie-Poisson structure}
Let 
\begin{equation}
\label{e:h_0}
    h_0: \ao^*\to\R, \quad 
    ([\Pi],q)^{\top} \mapsto \by{1}{2}\vv<([\Pi],q)^{\top} , [\mu_0^P]^{-1}([\Pi],q)^{\top}>
\end{equation}
denote the natural kinetic energy Hamiltonian.
The corresponding Lie-Poisson equations are 
\begin{align}
    \dot{[\Pi]} 
    &= -\ad(u)^*[\Pi] - X\diamond q 
    \label{e:lp1}
    \\
    \dot{q}
    &= -\rho^u(q) - \ad(X)^*q
    \label{e:lp2}\\
     \left(\begin{matrix} 
    u \\ 
    X
    \end{matrix}\right)
    &=   [\mu_0^P]^{-1}
    \left(\begin{matrix} 
    [\Pi] \\ 
    q 
    \end{matrix}\right)
    \label{e:lp3}
\end{align}
where the notation for $\rho$ is explained in Section~\ref{sec:SD_struc}.
The diamond $\diamond: \gau\times\gau^*\to\du^*$ is a bilinear map defined by
\begin{equation}
    \label{e:diamond}
    \vv<X\diamond q, u>_1 
    = \vv<q, \rho^u X>_2 
    = \vv<q, \dd{t}{}|_0 X(\exp( t u))>_2
    = \int_M \vv<q, L_u X>_3\,\vol_M
    = \int_M \vv<q, \nabla_u X>_3\,\vol_M
\end{equation}
where $\vv<.,.>_i$ for $i=1,2,3$ stands for the duality pairing on $\du^*\times\du$, $\gau^*\times\gau$, $\ko^*\times\ko$, respectively. 
Since $\ko$ is finite dimensional, the Lie derivative $L_u$ is applied component wise, and it coincides with the covariant derivative $\nabla_u$ because $\mu_0^M = \vv<.,.>$ is the Euclidean inner product. 
The Poisson structure on $\aut_0^*$ that underlies these Lie-Poisson equations is presented in Appendix~\ref{app:GM}. 

\begin{remark}[Charged fluid]\label{rem:YM_fluid}
The system~\eqref{e:lp1}, \eqref{e:lp2}, \eqref{e:lp3} is the Lie-Poisson version of \cite[Section~4.2]{GBTV13}, with the difference that \cite{GBTV13} are more general in the sense that they allow inner products that involve positive symmetric differential operators $Q_1: \X(M)\to\X(M)^*$ and $Q_2: \gau\to\gau^*$. We will need this generality in \eqref{e:Clp5}, \eqref{e:Clp6} below.  These equations describe the motion of an incompressible fluid consisting of charged particles in an external Yang-Mills field $\M = \textup{Curv}^{A_0} = dA_0 + \by{1}{2}[A_0,A_0]$. If $K$ is abelian, then $\M = dA_0$ is a magnetic field. Equation~\eqref{e:lp2} says that charge is conserved along the flow.  
\end{remark}

\subsection{Lie-Poisson approach: controlling the conserved quantity}
Let from now on $\nu = [\Pi]\in\du^*$ and define $[A_0^*]: \gau^*\to\du^*$, $q\mapsto[A_0^*q]$. 
The cotangent lifted right action by $\A$ on $T^*\A$ induces an equivariant momentum map
\begin{align}\label{e:momap1}
J: T^*\A = \D\times\G\times\du\times\gau &\to\du^*\times\gau^*,
(\phi,g,\nu,q)^{\top} \mapsto (\Ad(\phi,g)^{\top})^*(\nu,q)^{\top} 
\end{align}
where 
\begin{align*}
    \left<
    J\left(
    \left(\begin{matrix}
        \phi\\
        g
        \end{matrix}
        \right),
        \left(\begin{matrix}
        \nu\\
        q
        \end{matrix}
        \right)
        \right),
        \left(\begin{matrix}
        u\\
        X
        \end{matrix}
        \right)
    \right>
    &= 
    \left<
     \left(\begin{matrix}
        \nu\\
        q
        \end{matrix}
        \right),
    \left(\begin{matrix}
        \Ad(\phi)u\\
        \rho^{\phi^{-1}}\Big(
        TR^{g^{-1}}.Tg.u + \Ad(g)X
        \Big)
        \end{matrix}
        \right)
     \right> \\
     &=
        \left<
     \left(\begin{matrix}
        \Ad(\phi)^*\nu + 
        (Tg)^*(TR^{g^{-1}})^*\rho^{\phi}(q)\\
        \Ad(g)^*\rho^{\phi} (q)
        \end{matrix}
        \right),
        \left(\begin{matrix}
        u\\
        X
        \end{matrix}
        \right)
     \right> 
\end{align*}
The momentum map $J_{\mathcal{G}}$ with respect to the cotangent lifted $\G$ action is therefore given by 
\begin{equation}
 J_{\mathcal{G}}(\phi,g,\nu,q)
 = \pr_2 J(\phi,g,\nu,q)
 = (\Ad(g)^*\circ\rho^{\phi})(q)
\end{equation}
where $\pr_2: \du^*\times\gau^*\to\gau^*$ is the projection.

As in \eqref{e:C_k} we want to do work in the direction of the internal symmetry. Thus we consider a linear operator
\begin{equation}
    C: \du^*\to\gau^*.
\end{equation}


Define $F: T^*\A = \A\times\ao^*\to\gau^*$ through
\begin{align}
    \label{e:forceF1}
    F(\phi,g,\nu,q)
    &= \Ad(g)^*\rho^{\phi}\Big(
    (\rho^u+\ad(X)^*)C\nu - C(\ad(u)^*\nu + X\diamond q)\Big) \\
    &=
    \pr_2\left(
    \Ad\left(\begin{matrix} 
    \phi \\ 
    g
    \end{matrix}\right)^*
    \left(
    \ad\left(\begin{matrix} 
    u \\ 
    X
    \end{matrix}\right)^*
    \left(\begin{matrix} 
    0 \\ 
    C\nu
    \end{matrix}\right)
    -
    \left(\begin{matrix} 
    0 \\ 
    C(\ad(u)^*\nu + X\diamond q)
    \end{matrix}\right)
    \right)
    \right)
\end{align}
where $(u,X)=[\mu_0^P]^{-1}(\nu,q)$.
Assume that $(\phi_t,g_t,\nu_t,q_t)\in\D\times\G\times\du^*\times\gau^*$ is a curve such that:
\begin{enumerate}
    \item 
    $(\phi_0,g_0) = (e,e)$ and $(\dot{\phi},\dot{g}) = TR^{(\phi,g)}(u,X)$ where $R$ is the $\rho$-dependent semi-direct product right multiplication on $\A$ and $(u,X) = [\mu_0^P]^{-1}(\nu,q)$.
    \item
    \eqref{e:lp1} holds.
\end{enumerate}
The dynamical equation for $q$ is now given by applying the force $F$ such that 
\begin{equation}
    \dd{t}{}J_{\mathcal{G}}(\phi,g,\nu,q)
    = -F(\phi,g,\nu,q).
\end{equation}
Because of \eqref{e:lp1} 
it follows that 
$F(\phi,g,\nu,q) = \frac{\del}{\del t}\Ad(g)^*\rho^{\phi}(C\nu)$,
whence we obtain a new conserved quantity, namely 
\begin{equation}
    \label{e:P0}
    J_{\mathcal{G}}(\phi,g,\nu,q) + \pr_2\Ad(\phi,g)^*(0,C\nu)
    = \Ad(g)^*\rho^{\phi}(q + C\nu)
    = p_0 
    = \textup{const}.
\end{equation}
As in Section~\ref{sec:rigidB}, we introduce a new variable $p$:
\begin{equation}\label{e:PC}
    p_t := (\rho^{\phi_t^{-1}}\circ\Ad(g_t^{-1})^*)p_0
    = q_t + C\nu_t 
\end{equation}
where we have added the subscript $t$ to highlight the time dependence.
Therefore, the control law following from the force~\eqref{e:forceF1} is 
\begin{equation}
    \label{e:controlF1}
    q_t = p_t - C\nu_t
\end{equation}
which should be compared with equation~\eqref{e:controlPk}.
The resulting equations of motion are 
\begin{align}
    \dot{\nu} 
    &= -\ad(u)^*\nu - X\diamond q 
    \label{e:Clp1}
    \\
    \dot{q}
    &= -\rho^u(q) - \ad(X)^*q
        - (\rho^u)^*C\nu  -  \ad(X)^*C\nu 
        + C\Big(\ad(u)^*\nu + X\diamond q\Big) 
    \label{e:Clp2}\\
  \left(\begin{matrix} 
    u \\ 
    X
    \end{matrix}\right)
    &=   [\mu_0^P]^{-1}
    \left(\begin{matrix} 
    \nu \\ 
    q 
    \end{matrix}\right)
    \label{e:lp2a}
\end{align}
 We think of $(\nu,q)$  as the physical variables which are controlled, while $(\nu,p)$ are abstract variables which should be described by means of a Hamiltonian (Lie-Poisson) system.
In terms of the latter, equations~\eqref{e:Clp1} and \eqref{e:Clp2} can be rewritten as
\begin{align}
    \dot{\nu} 
    &= -\ad(u)^*\nu - X\diamond(p-C\nu)
    \label{e:Clp3}
    \\
    \dot{p}
    &= -\rho^u(p) - \ad(X)^*p
    \label{e:Clp4}
\end{align}
where \eqref{e:lp2a} is changed to
\begin{equation}
\label{e:muC1}
    \left(\begin{matrix} 
    u \\ 
    X
    \end{matrix}\right)
    =
    [\mu_0^P]^{-1}\left(\begin{matrix} 
    \nu \\ 
    p - C\nu
    \end{matrix}\right)
    =
    [\mu_0^P]^{-1}
    \left(\begin{matrix} 
    1 & 0 \\ 
    -C & 1
    \end{matrix}\right)
    \left(\begin{matrix} 
    \nu \\ 
    p 
    \end{matrix}\right).
\end{equation}
Equation~\eqref{e:Clp3} is the closed loop equation associated to \eqref{e:lp1}, \eqref{e:lp3} and the feedback control \eqref{e:controlF1}.

\begin{remark}\label{rem:LP_C}
The above equations \eqref{e:Clp3}, \eqref{e:Clp4} and \eqref{e:muC1} almost look like a Lie-Poisson system. The problem is that the matrix
\[
   [\mu_0^P]^{-1}
    \left(\begin{matrix} 
    1 & 0 \\ 
    -C & 1
    \end{matrix}\right)
\]
is not symmetric, thus this term does not give rise to a metric operator. If the diamond term in \eqref{e:Clp3} was not present, we could introduce a new variable $\tilde{p}$ to account for this asymmetry. This was the approach of Section~\ref{sec:rigidB}. See Remarks~\ref{rem1} and \ref{rem3}. However, because of the diamond term, this does not work for the semi-direct product structure. 
\end{remark}

\subsection{Lie-Poisson approach: controlling the conserved quantity with $T$ and $C$}\label{sec:LP_CT}
As in \eqref{e:C_k} we want to do work in the direction of the internal symmetry. Thus we consider a linear map 
\begin{equation}
    C: \du^*\to\gau^*.
\end{equation}
Let $T: \gau^*\to\gau^*$ be an isomorphism. 

Define $F: \ao^* = \du^*\times\gau^*\to\gau^*$ through
\begin{align}
    \label{e:forceF2}
    F(\nu,q)
    := 
    &- C\Big(\ad(u)^*\nu+X\diamond q\Big) 
    + T\Big(\ad(X)^*+\rho^u\Big)(T^{-1}C\nu)\\
    &+ \Big(T(\ad(X)^*+\rho^u)T^{-1} - \ad(X)^*-\rho^u\Big)q
    \notag
\end{align}
where $(u,X)^{\top}=[\mu_0^P]^{-1}(\nu,q)^{\top}$.

\begin{proposition}\label{prop:TCforce}
Let $(\phi_t,g_t,\nu_t,q_t)\in T^*\A$ be a curve such that $(\phi_t,\nu_t)$ solves \eqref{e:lp1} subject to \eqref{e:lp3} and $(\dot{\phi},\dot{g}) = TR^{(\phi,g)}(u,X)$. Let $(\phi_0,g_0)=(e,e)$.
The following are equivalent.
\begin{enumerate}
    \item 
    With $\by{Dq}{dt} = \dot{q} + \ad(X)^*q + \rho^u q$,
    \begin{equation}
    \label{e:TCforce}
        \frac{Dq}{dt} = -F(\nu,q).
    \end{equation}
    \item
    The following quantity is conserved: 
    \begin{equation}
        \label{e:PC2}
        \Ad(g_t)^*\rho^{\phi_t}( T^{-1}q_t + T^{-1}C\nu_t )
        = p_0 
        = \textup{const}.
    \end{equation}
\item
    The quantity $p_t := T^{-1}q_t + T^{-1}C\nu_t$
    satisfies
    \begin{equation}
    \label{pdot_TC}
        \dot{p} = -\rho^u p - \ad(X)^*p.
    \end{equation}
\end{enumerate}
\end{proposition}

\begin{proof}
(2) and (3) are equivalent, since $p_t = T^{-1}q_t + T^{-1}C\nu_t$ and 
\[
 \dd{t}{}\Ad(g_t)^*\rho^{\phi_t} p_t
 = \Ad(g_t)^*\rho^{\phi_t}(\dot{p}_t + \rho^{u_t} p_t + \ad(X_t)^*p_t).
\]
Further, (1) implies (3) because, using \eqref{e:lp1},
\begin{align*}
    \dot{p}
    &= T^{-1}\dot{q} + T^{-1}C\dot\nu
    = T^{-1}\Big(\frac{Dq}{dt} - \ad(X)^*q - \rho^u q + C\dot{\nu}\Big)
    \\
    &= T^{-1}\Big(
        - C\dot{\nu}
        - T(\ad(X)^*+\rho^u)(T^{-1}C\nu)\\
    &\phantom{===}
        - (T(\ad(X)^*+\rho^u)T^{-1} - \ad(X)^* - \rho^u)q\\
    &\phantom{===}
        - \ad(X)^*q - \rho^u q + C\dot{\nu}\Big)\\
    &= 
    -(\ad(X)^*+\rho^u)T^{-1}C\nu
    -(\ad(X)^*+\rho^u)T^{-1}q
\end{align*}
and the converse direction follows by calculating $\dot{q}$ from this equation.
\end{proof}

Thus the control law associated to exerting the force \eqref{e:TCforce} is
\begin{equation}
    \label{e:PC3}
    q_t = Tp_t - C\nu_t
\end{equation}
where, as in \eqref{e:PC}, $p_t : = \rho^{\phi_t^{-1}}\Ad(g_t^{-1})^*p_0$.
The controlled equations of motion follow from \eqref{e:lp1}, \eqref{pdot_TC} and \eqref{e:lp3}, and are
\begin{align}
    \dot{\nu} 
    &= -\ad(u)^*\nu - X\diamond (Tp - C\nu)  
    \label{e:Clp3a}
    \\
    \dot{p}
    &= -(\rho^u + \ad(X)^*)p
    \label{e:Clp4a}\\
      \left(\begin{matrix} 
    u \\ 
    X
    \end{matrix}\right)
    &=
    [\mu_0^P]^{-1}
    \left(\begin{matrix} 
    \nu \\ 
    Tp - C\nu
    \end{matrix}\right)
    =
    [\mu_0^P]^{-1}
    \left(\begin{matrix} 
    1 & 0 \\ 
    -C & T
    \end{matrix}\right)
    \left(\begin{matrix} 
    \nu \\ 
    p 
    \end{matrix}\right).
    \label{e:lp-fb-law}
\end{align}
Since $\mu_0^P = \mu^{KK}(\mu_0^M,\I_0,A_0)$, Lemma~\ref{pr:metric1} implies 
\begin{equation}
\label{e:muCT1}
    [\mu_0^P]^{-1}
    \left(\begin{matrix} 
    1 & 0 \\ 
    -C & T
    \end{matrix}\right)
    =
    \left(\begin{matrix} 
    [\mu_C^M]^{-1} & -[\mu_0^M]^{-1}[A_0^*]T \\ 
    -A_0[\mu_C^M]^{-1} -\I_0^{-1}C & Q_0T
    \end{matrix}\right)
\end{equation}
where 
\begin{equation}
\label{e:mu_C^M}
     [\mu_C^M] := (1+[A_0^*]C)^{-1}[\mu_0^M]
\end{equation}
and  
\begin{equation}
    Q_0 := \I_0^{-1}+A_0[\mu_0^M]^{-1}[A_0^*]. 
\end{equation}
We assume that $1+[A_0^*]C$ is invertible.

\begin{remark}\label{rem:mu_C}
Note that we did not assume that $1+[A_0^*]C$ is a local operator. In fact, $C$ may involve a convolution, as in the example in Section~\ref{sec:shear}. Thus $[\mu_C^M]$ need not come from a Riemannian metric $\mu_C^M$ on $M$. Nevertheless, it still makes sense to require \eqref{e:muCT1} to define a Kaluza-Klein inner product on $\ao$ as in Lemma~\ref{pr:metric1}. This is the situation of \cite[Equ.~(6)]{GBTV13}, compare with Remark~\ref{rem:YM_fluid}. 
\end{remark}

Owing to Lemma~\ref{pr:metric1}, for the expression \eqref{e:muCT1} to be a Kaluza-Klein inner product on $\ao$, with the connection form
\begin{align}
     \label{e:TCcond1}
    A_C
    := A_0  +\I_0^{-1}C [\mu_C^M]
     = A_0 + \I_0^{-1}C(1+[A_0^*]C)^{-1}[\mu_0^M],
\end{align}
the following have to hold:
\begin{align}
    \label{e:TCcond1b}
    A_C[\mu_C^M]^{-1}
    &= T^*A_0[\mu_0^M]^{-1}\\
    \label{e:TCcond2}
    \I_C^{-1}  
    &= \I_0^{-1}T - (A_C-A_0)[\mu_C^M]^{-1}[A_C^*]
\end{align}
which is equivalent to
\begin{align}
    \label{e:TCcond1a}
    [A_0^*]T 
    &= [\mu_0^M][\mu_C^M]^{-1}[A_C^*]
    = (1+[A_0^*]C)[A_C^*] \\
    \label{e:TCcond2a}
    \I_C^{-1}
    &= \I_0^{-1}(T-C[A_C^*])
\end{align}
where $T-C[A_C^*]$ is assumed to be invertible.
Recall that we assume throughout that $(\mu_0^M,\I_0,A_0)$ is a Kaluza-Klein triple which means, in particular, that $\I_0$ is $\Ad(K)$-invariant. We obtain: 

\begin{theorem}
\label{thm:1}
Assume $1+[A_0^*]C$, $T$ and $T-C[A_C^*]$ are invertible, and that $T$ satisfies \eqref{e:TCcond1a}. 
Let $L_1 := (1+[A_0^*]C)^{-1}$ and $L_2 := (T-C[A_C^*])^{-1}$ and assume further that
\begin{equation}
    L_1^* = [\mu_0^M]^{-1}L_1[\mu_0^M],
    \quad
    L_2^* = \I_0^{-1}L_2\I_0
\end{equation}
and $\Ad(k)^*L_2\I_0\Ad(k) = L_2\I_0$ for all $k\in K$.
Let 
 \begin{equation}
    \label{e:muC}
    [\mu_C^P] = \mu^{KK}([\mu_C^M],\I_C,A_C).
\end{equation}
be the Kaluza-Klein inner product on $\ao$ associated to $[\mu_C^M] = L_1[\mu_0^M]$, $\I_C = L_2\I_0$ and $A_C$ as defined by  \eqref{e:TCcond1}. Then 
\begin{equation}
    [\mu_C^P]^{-1} 
    =  
    [\mu_0^P]^{-1}
    \left(\begin{matrix} 
    1 & 0 \\ 
    -C & T
    \end{matrix}\right)
\end{equation}
and equations~\eqref{e:Clp3a}, \eqref{e:Clp4a}, \eqref{e:lp-fb-law} can be written as a forced Lie-Poisson system
\begin{align}
    \dot{\nu} 
    &= -\ad(u)^*\nu - X\diamond p  + f 
    \label{e:Clp5}
    \\
    \dot{p}
    &= -\rho^u(p) - \ad(X)^*p
    \label{e:Clp6}\\
      \left(\begin{matrix} 
    u \\ 
    X
    \end{matrix}\right)
    &=
    [\mu_C^P]^{-1}\left(\begin{matrix} 
    \nu \\ 
    p 
    \end{matrix}\right)
\end{align}
with a force term
\begin{equation}
\label{e:lpf}
    f :=  - X\diamond ((T-1)p - C\nu)  
\end{equation}
\end{theorem}
 
\begin{remark}
An applied force $F$ on the internal variable $q$ results, via the feedback law~\eqref{e:PC3}, in an equivalent system with a force $f$ acting on the fluid flow variable. 
\end{remark} 

\begin{remark}
The map $T$ replaces the need for introducing a new variable $\tilde{p}$. Compare with Remarks~\ref{rem1} and \ref{rem3}. As explained in Remark~\ref{rem:LP_C}, the $\tilde{p}$-construction does not work for fluid dynamical systems. 
\end{remark}

\subsection{Momentum $p=0$}
Stabilization in Section~\ref{sec:RBstab} was achieved by considering the case of zero momentum. For the general system~\eqref{e:Clp5}, $p=0$ yields
\begin{align}
    \dot{\nu} 
    &= -\ad(u)^*\nu  + f 
    \label{e:Clp7}
    ,\qquad
    u =
    [\mu_C^M]^{-1}\nu  
\end{align}
with a force term
\begin{equation}
    f =  - \Big(A_C[\mu_C^M]^{-1}\nu\Big) \diamond \Big(C\nu\Big) 
\end{equation}

\subsection{$C=\gamma\I_0 A_0[\mu_0^M]^{-1}$ and $p=0$}
Let $\gamma\in\R$ be a control parameter and define the control law~\eqref{e:PC3} by 
\begin{equation}
\label{e:C_IAmu}
    C = \gamma\I_0 A_0[\mu_0^M]^{-1}.
\end{equation}

\begin{remark}
Note that equation~\eqref{e:C_k1} can be written in this form with $k = -\gamma i_3 I_3^{-1}$.
\end{remark}

Using definition \eqref{e:diamond},  the force $f$ on $v\in\du$ is, with $\nu\in\du^*$ and $Y=A_0[\mu_0^M]^{-1}\nu$,
\begin{align}
    f(v) &= -\int_M\vv<C\nu,\nabla_v\Big(A_C[\mu_C^M]^{-1}\nu\Big)>\,dx 
    \notag\\
    &= 
    -\int_M\vv<\gamma \I_0 A_0[\mu_0^M]^{-1}\nu,
        \nabla_v\Big((A_0[\mu_0^M]^{-1}(1+[A_0^*]C)+\I_0^{-1}C)\nu\Big)> \,dx  
        \label{e:fis0}\\
    &= -\gamma\int_M\vv<\I_0Y, 
        \nabla_v\Big((1+\gamma A_0[\mu_0^M]^{-1}[A_0^*]\I_0+\gamma)Y\Big)>\,dx
        \notag \\
    &= -\gamma\int_M\vv<\I_0Y, 
        \nabla_v\Big((1+\gamma Q_0\I_0 )Y\Big)>\,dx 
        \notag \\
    &=
    -\gamma_M\int\textup{div}\Big(\vv<\I_0 Y,Y>v\Big)\,dx
    -\gamma^2\int_M\vv<\I_0Y, 
        \nabla_v \Big(Q_0\I_0 Y\Big)>\,dx 
        \notag 
\end{align}
Since $\I_0$ is constant it follows that $f(v)=0$, if $\nabla_v (Q_0\I_0 Y) = \lam\nabla_v Y$ for some $\lam\in\R$.

\subsection{$C=\gamma R^{-1}\I_0 A_0[\mu_0^M]^{-1}$ and $p=0$}\label{sub:CTM}
Now we shall modify the control so that $f=0$.
Since $\I_0$ is constant,
equation~\eqref{e:fis0} implies that $f = 0$ if 
\begin{equation}\label{e:CCC}
 \lam\I_0^{-1} C 
 = \gamma\Big(A_0[\mu_0^M]^{-1}[A_0^*]C + \I_0^{-1}C + A_0[\mu_0^M]^{-1}\Big) 
\end{equation}
for control parameters $\lam$ and $\gamma$. Clearly, $\lam$ and $\gamma$ are not independent, and it will be convenient to choose $\lam := \gamma+1$. Let
\begin{equation}
    \label{e:Cgamma}
    C := \gamma R^{-1}\I_0 A_0 [\mu_0^M]^{-1}
\end{equation}
where 
\begin{equation}
    \label{e:DefR}
    R := 1 - \gamma \I_0 A_0 [\mu_0^M]^{-1} [A_0^*]
\end{equation}
and $\gamma$ is chosen so that $R^{-1}$ exists. If $R^{-1}$ exists, so that $C$ is well-defined, we have $R^{-1} = 1 + C[A_0^*]$. It follows that $C$ satisfies \eqref{e:CCC}. 

We recall that we assume throughout that $(\mu_0^M,\I_0,A_0)$ is a Kaluza-Klein triple which means, in particular, that $\I_0$ is $\Ad(K)$-invariant. 

\begin{theorem}
\label{thm:M_CT}
Define $A_C$ by \eqref{e:TCcond1}. 
Assume that $\gamma$ is sufficiently small so that $R$, $1+[A_0^*]C$ and
\begin{align}
    \label{e:DefT}
    T := 1 + \gamma + C[A_C^*] = 1 + \gamma + (1+\gamma)C[A_0^*]
\end{align}
are invertible.
Define $[\mu_C^M]$ and $\I_C$ by \eqref{e:mu_C^M} and \eqref{e:TCcond2a}, respectively. 
Then the following are true:
\begin{enumerate}[\up (1)]
\item 
$T$ satisfies \eqref{e:TCcond1a}.
\item
$\I_C = (1+\gamma)^{-1}\I_0$ is invertible and $\Ad(K)$-invariant. 
Moreover,
$[\mu_C^M]^* = [\mu_C^M]$ and $\I_C^* = \I_C$. 
\item
The feedback control system \eqref{e:Clp5}, at momentum  $p=0$, is 
\begin{align}
  \label{e:lpa}
    \dot{\nu} 
    &= -\ad(u_C)^*\nu, 
    \qquad
    u_C =
    [\mu_C^M]^{-1}\nu  
\end{align}
which is a Lie-Poisson system with respect to the Hamiltonian 
\begin{equation}
    \label{e:lpc}
    h_M^C: \du^*\to\R,\quad
    \nu\mapsto 
    \by{1}{2}\vv<\nu,[\mu_C^M]^{-1}\nu> 
    =  \by{1}{2}\int\vv<\nu,[\mu_C^M]^{-1}\nu>\,\textup{vol}_M
\end{equation}
\item
If $\nu_e$ is an equilibrium of the Lie-Poisson system associated to the uncontrolled ($C=0$) Hamiltonian $h_M^0$ and $C\nu_e = 0$, then $\nu_e$ is also an equilibrium for \eqref{e:lpa}.
\end{enumerate}
\end{theorem}

\begin{proof}
To see that $T$ satisfies \eqref{e:TCcond1a}, note that $[A_0^*]R = (1+A_0^*C)^{-1}[A_0^*]$, which implies
\begin{align*}
    (1+[A_0^*]C)[A_C^*]
    &= [A_0^*] + [\mu_C^M]C^*\I_0^{-1} + [A_0^*]C ( [A_0^*] + [\mu_C^M]C^*\I_0^{-1})\\
    &= [A_0^*]
        + \gamma(1+[A_0^*]C)^{-1}[A_0^*]R^{-1}
        + [A_0^*] C [A_0^*]
        + \gamma[A_0^*]C(1+[A_0^*]C)^{-1}A_0^*R^{-1} \\
    &= [A_0^*]\Big(
        1+\gamma C(1+[A_0^*]C)^{-1}[A_0^*]\Big)R^{-1}
        + \gamma (1+[A_0^*]C)^{-1}[A_0^*] R^{-1} \\
    &= [A_0^*]\Big(R^{-1} + \gamma + \gamma C[A_0^*]\Big)
    = [A_0^*]\Big(1 + C[A_0^*] + \gamma + \gamma C[A_0^*]\Big)
    = [A_0^*]T.
\end{align*}

Since $R^* = \I_0^{-1}R\I_0$, it follows that 
\begin{align*}
    \Big([\mu_C^M]^{-1}\Big)^*
    &= 
    \Big([\mu_0^M]^{-1}(1+[A_0^*]C)\Big)^*
    =
    (1+\gamma [\mu_0^M]^{-1}[A_0^*]\I_0(R^{-1})^*A_0)[\mu_0^M]^{-1}\\
    &=
    [\mu_0^M]^{-1}(1 + \gamma [A_0^*]R^{-1}\I_0 A_0 [\mu_0^M]^{-1})
    = [\mu_C^M]^{-1}.
\end{align*}

Since $C$ satisfies \eqref{e:CCC} by construction, it follows that  
\begin{align*}
    \gamma f(v) 
    &= -\gamma\int_M\vv<C\nu,\nabla_v\Big(A_C[\mu_C^M]^{-1}\nu\Big)>\,dx \\
    &= 
    -\gamma\int_M\vv< C \nu,
        \nabla_v\Big((A_0[\mu_0^M]^{-1}(1+[A_0^*]C)+\I_0^{-1}C)\nu\Big)> \,dx  
   \\
    &= -\int_M\vv<C\nu , \nabla_v\Big( (\gamma+1)\I_0^{-1}C\nu\Big)>\,dx \\
    &= -(\gamma+1)\int_M \textup{div}\Big(\vv<C\nu ,\I_0^{-1}C\nu>v\Big)\,dx
    = 0 
\end{align*}
Concerning the last point, note that $C\nu_e=0$ implies $[\mu_0^M]^{-1}\nu_e = [\mu_C^M]^{-1}\nu_e$.
\end{proof}


\begin{remark}
The map $T$ does not appear in the system \eqref{e:lpa}  and \eqref{e:lpc}, because we assume $p=0$ which, according to \eqref{e:PC3}, means the control is $q=-C\nu$. However, $T$ was used  in the construction of the forced Lie-Poisson system of Theorem~\ref{thm:1}, without which we could not have obtained the Lie-Poisson system \eqref{e:lpa}. Further, the force~\eqref{e:forceF2}, which contains the physical interpretation of the control $C$, does depend on $T$. 
\end{remark}

\begin{remark}\label{rem:CfromA}
Once $A_0$ is specified, the only free parameter that remains is $\gamma$. The control $C$ is completely determined by the requirement that the force $f$ should vanish and $T$ is, in turn, fixed by \eqref{e:TCcond1a}. Thinking of the rigid body example, this makes physical sense: Once the rotor $A_0$ is attached to the third axis, the only quantity left to control is the rotor's speed. Now in order to obtain a conserved quantity, the rotor's speed has to be constant. This constant is the parameter $\gamma$ and Proposition~\ref{prop:RBstab} says how to choose $\gamma = -i_3^{-1} I_3 k$ so that rotation about the middle axis is stable. While the force $f$ is due to the diamond term and therefore absent in the rigid body example, the fact that the only free choice concerns the parameter $\gamma$ is a property that seems fundamental to the controlled Hamiltonian method.
\end{remark}

\section{Stabilization}\label{sec:stab}
Let the notation be as in Section~\ref{sec:FB_IS}. Assume that $\nu_e$ is an equilibrium of the uncontrolled Lie-Poisson system 
\begin{equation}
\label{e:E0}
    \dot{\nu} 
    = -\ad([\mu_0^M]^{-1}\nu)^*\nu
\end{equation}
in $\du^*$. 
This equation is Lie-Poisson with respect to the standard Poisson bracket $\{.,.\}$ inherited from the canonical symplectic form on $T^*\D$. 
It is the Euler equation for the Hamiltonian $h_0^M: \du^*\to\R$, $\nu\mapsto\by{1}{2}\vv<\nu,[\mu_0^M]^{-1}\nu>$. 

Since the controlled system \eqref{e:lpa}  and \eqref{e:lpc} is Lie-Poisson with respect to the same bracket $\{.,.\}$, the set of Casimir functions remains unchanged. Moreover, the symplectic leaves depend only $\{.,.\}$ but not on the dynamical equation. Thus, if the symplectic leaf  passing through $\nu_e$ consists only of this one point, then $\nu_e$ is also an equilibrium of \eqref{e:lpa}. 

Let us assume that $\nu_e$ is an \emph{unstable}  equilibrium of \eqref{e:E0}. The goal is to find a control $C$ (acting on the internal momentum variable $q$) such that $\nu_e$ is a stable equilibrium of \eqref{e:lpa}.  
To this end we consider some of the points of the stabilization algorithm of \cite{HMRW85}:
\begin{itemize}
\item[(A)]
Hamiltonian form: 
This is satisfied by construction since \eqref{e:lpa} is  Lie-Poisson.
\item[(C)]
First variation:
If $g_0$ is a constant of motion such that $\nu_e$ is a critical point of $h^M_0 + g_0$,
then $\nu_e$ may (as in the rigid body example in Section~\ref{sec:RBagain}) or may not (as in the shear flow example of Section~\ref{sec:shear}) be a critical point of $h_C^M + g_0$.
In the latter case the potential $A_0$ and the control $C$ need to be constructed so that there exists a constant of motion $g_C$ such that $\nu_e$ is a critical point of $h_C^M+g_C$.
\item[(D1)]
Formal stability: With $g=g_0$ or $g=g_C$, show that the second variation $D^2(h_C^M+g)(\nu_e)$ is positive or negative definite. Alternatively, one may consider the inclusion $\iota: \mathcal{O}\to\du^*$ of the coadjoint orbit $\mathcal{O}$ through $\nu_e$, and show that $D^2(\iota^*h_C^M)(\nu_e)$ is definite. 
\end{itemize}

\begin{remark}
Starting from a general form of the control $C$, the idea is that the above steps should yield conditions on $C$ such that $\nu_e$ is a stable equilibrium of \eqref{e:lpa}. This presupposes the existence of internal symmetries $\G$ on which $C$ can act (via the force \eqref{e:forceF2}). It should be noted that this process only allows to conclude stability with respect to variations in the $\nu$-variables.  
\end{remark}


The following is a version of \cite[Equ.~(44)]{A66}:\footnote{In \eqref{e:SV} we do not have the factor $\by{1}{2}$ of  \cite[Equ.~(44)]{A66} because we use the convention of \cite{HMRW85} for the definition of the second variation.} 

\begin{proposition}[Second variation]\label{prop:SV}
Let $\nu_e$ be an equilibrium of \eqref{e:lpa}, let $\mathcal{O}$ be the coadjoint orbit through $\nu_e$ and let $\iota: \mathcal{O}\to\du^*$ be the inclusion. Then the second variation of $\iota^*h_C^M$ at $\nu_e$ is given by 
\begin{equation}
    \label{e:SV}
    D^2(\iota^*h_C^M)(\nu_e)(\delta\nu,\delta\nu)
    = 
    \vv<\delta\nu, [\mu_C^M]^{-1}\delta\nu + [v,[\mu_C^M]^{-1}\nu_e]>
\end{equation}
where $\delta\nu = \ad(v)^*\nu_e\in T_{\nu_e}\mathcal{O}$ with $v\in\du$.
\end{proposition}

\begin{proof}
Since $\exp(tv)$ and $\exp(sv)$ commute,
\begin{align*}
    &D^2(\iota^* h_C^M)(\nu_e)(\delta\nu,\delta\nu)
    = 
    \frac{\del^2}{\del t^2}|_0
    (\iota^* h_C^M)(\Ad(\exp(tv))^*\nu_e)\\
    &=
    \frac{1}{2}\frac{\del}{\del s}|_0\frac{\del}{\del t}|_0
    \vv<\Ad(\exp(tv))^*\Ad(\exp(sv))^*\nu_e,
    [\mu_C^M]^{-1}\Ad(\exp(tv))^*\Ad(\exp(sv))^*\nu_e>\\
    &=
    -\frac{\del}{\del s}|_0
    \vv<\ad([\mu_C^M]^{-1}\Ad(\exp(sv)^*\nu_e)^*\Ad(\exp(sv)^*\nu_e,v>\\
    &=
    \vv<\ad(v)^*\nu_e
        ,[\mu_C^M]^{-1}\ad(v)^*\nu_e
         + [v,[\mu_C^M]^{-1}\nu_e]>.
\end{align*}
To show that \eqref{e:SV} depends only on $\delta\nu$ and not on $v$, one proceeds exactly as in \cite{A66}: Suppose $\delta\nu=\ad(v)^*\nu_e=\ad(w)^*\nu_e$, and use the Jacobi identity and $\ad([\mu_C^M]^{-1}\nu_e)^*\nu_e = 0$ to show that
\begin{align*}
    \vv<\delta\nu, [v-w,[\mu_C^M]^{-1}\nu_e]>
    &=
    \vv<\nu_e, 
        [[v,v-w],[\mu_C^M]^{-1}\nu_e]
        + [v-w,[v,[\mu_C^M]^{-1}\nu_e]>\\
    &= \vv<\ad(v-w)^*\nu_e, [v,[\mu_C^M]^{-1}\nu_e]> 
    = 0.
\end{align*}
\end{proof}

\section{The rigid body with a rotor in  the context of Theorem~\ref{thm:M_CT}}\label{sec:RBagain}
Let us apply Theorem~\ref{thm:M_CT} and the algorithm of Section~\ref{sec:stab} to the rigid body example of Section~\ref{sec:rigidB}. 
This is along the lines of \cite[Equ.~(3.1C)]{HMRW85}. 
In the notation of Section~\ref{sec:FB_IS} we have now $\D = \SO(3)$ and $\G=S^1$. The connection is $A_0u = \vv<e_3,u> = u_3$, the metric is given in \eqref{e:muh} as 
\begin{equation}
    \mu_0^M = \mu_0^S = 
    \diag{\lam_1,\lam_2,I_3}
\end{equation}
where $\lam_j = I_j + i_j$, 
and the inertia tensor is $\I_0 X = i_3 X$. As in Section~\ref{sec:rigidB}, we assume the ordering $I_1>I_2>I_3$. 
The corresponding Hamiltonian function on $\R^3$, viewed as the dual to the Lie algebra of $\SO(3)$, is $h_0: \R^3\to\R$, $\Pi\mapsto\by{1}{2}\vv<\Pi,(\mu_0^S)^{-1}\Pi>$. 
Since the Kaluza-Klein data $\mu_0^M$, $\I_0$ and $A_0$ are constant in the configuration space variable, we can apply formula \eqref{e:Cgamma} with $R=1$ and define the control as
\begin{equation}
    C: \R^3\to\R,
    \quad
    \Pi\mapsto 
    \gamma\I_0 A_0 (\mu_0^S)^{-1}\Pi
    = i_3 I_3^{-1}\gamma \Pi_3
\end{equation}
which coincides with \eqref{e:C_k1} if we put 
\begin{equation}
    \label{e:k-gamma}
    k = -i_3 I_3^{-1}\gamma.
\end{equation}
The Casimirs of the Poisson algebra $C^{\infty}(\R^3)$ are of the form
\[
 g_{\var}(\Pi) = \var(\by{|\Pi|^2}{2})
\]
for $\var\in C^{\infty}(\R)$. 
Consider the equilibrium $\Pi_e = (0,1,0)^{\top}$. This is unstable for the Lie-Poisson system associated to $h_0$. To relate $\Pi_e$ to a Casimir function, note that 
\[
 D(h_0 + g_{\var})(\Pi_e)\delta\nu
 = \vv<(\mu_0^S)^{-1}\Pi_e + \var'(\by{|\Pi_e|^2}{2})\Pi_e,\delta\nu>
\]
which implies the condition  $\var'(\by{1}{2}) = -\lam_2^{-1}$. Now, observe that 
\begin{equation}
    C \Pi_e = 0. 
\end{equation}
With equation~\eqref{e:mu_C^M} this yields 
\[
 Dh_0(\Pi_e)
 = (\mu_0^S)^{-1}\Pi_e
 = (\mu_C^S)^{-1}\Pi_e
 = Dh_C(\Pi_e)
\]
where $h_C$ is, according to \eqref{e:lpc}, given by $h_C: \R^3\to\R$, $\Pi\mapsto\by{1}{2}\vv<\Pi,(\mu_C^S)^{-1}\Pi>$ and
$\mu_C^S = (1+A_0^*C)^{-1}\mu_0^S = (1-k e_3e_3^{\top})^{-1}\mu_0^S = \eqref{e:muS2}$ as defined by \eqref{e:mu_C^M}.

\begin{remark}
In step (C) in the stability algorithm this corresponds to the case where we do not have to change the constant of motion. Thus $g_{\var}$ can be used for the stability analysis. 
\end{remark}

Now one proceeds with the stability analysis by calculating the second variation $D^2(h_C + g_0)(\Pi_e)$. One finds the stability condition
\begin{equation}
    \label{e:stabRBalg}
    -i_3 I_3^{-1} \gamma > 1 - I_3 \lam_2^{-1}
\end{equation}
which
makes $D^2(h_C + g_0)(\Pi_e)$ negative definite and, due to \eqref{e:k-gamma}, coincides with Proposition~\ref{prop:RBstab}.

\begin{remark}\label{rem:T_RBagain}
Contrary to Section~\ref{sec:rigidB}, specifically equation~\eqref{e:HamLoop1}, we did not need to introduce a new momentum value $\tilde{p}$. This is because of the map $T$. However, to correctly apply Theorem~\ref{thm:M_CT}, it should be checked that $T$ is invertible. A calculation indeed shows that $T = (i_3-k\lam_3)/i_3$, which is consistent with \eqref{e:pk_tilde}.  
\end{remark}

\section{Formal stability of two-dimensional fluids in an external field}\label{sec:stab_2D}
Consider now a compact domain $M\subset\R^2$ with smooth boundary $\del M$. Assume that $M$ is filled with an ideal  incompressible fluid which consists of charged particles. The fluid is subject to an external field and we wish to control the charge of the fluid in a domain dependent manner. The external field is modelled as the curvature of a connection on a trivial principal bundle $P=M\times K$ over $M$.

\subsection{Formal stability}\label{sec:formal_stab_2D}
Let $P=M\times K$, $\mu_0^M$ is the Euclidean inner product, $\I_0=1$, $A_0\in\Om^1(M,\ko)$ and $\mu_0^P$ is given by \eqref{e:mu0}.
Let the control be given by \eqref{e:Cgamma} and consider the feedback system \eqref{e:lpa} and \eqref{e:lpc}. 

The stream function of a vector field $v\in\du$ is a function $\psi$ such that $\psi$ is constant on the boundary and 
\[
 \nabla^s\psi = (-\del_y\psi,\del_x\psi)^{\top} = v.
\]
The vorticity of an element $\nu = [\Pi]\in\du^*$ is the function $\om^{\nu} = *d\Pi$ where $*$ is the Hodge star operator.

Note that vorticity is conserved by \eqref{e:lpa}. Indeed, with $u_C = [\mu_C^M]^{-1}(\nu)$
\begin{equation}
    \label{e:vor_cons}
    \dot{\om}^{\nu}
    = - * d(\ad(u_C)^*\Pi)
    = - * d(L_{u_C}\Pi)
    = - * L_{u_C}(d\Pi)
    = -L_{u_C}\om^{\nu}
\end{equation}
where $L_{u_C} = d i(u_C) + i(u_C) d$ is the Lie derivative along $u_C$.

\begin{proposition}[Second variation]\label{prop:SV2D}
Let $\nu_e$ be an equilibrium of the uncontrolled Euler equation~\eqref{e:E0} and of the controlled Euler equation~\eqref{e:lpa}, let $\mathcal{O}$ be the coadjoint orbit through $\nu_e$ and let $\iota: \mathcal{O}\to\du^*$ be the inclusion. 
Let $\psi_0$ and $\psi_C$ be the stream functions of $[\mu_0^M]^{-1}\nu_e$ and $[\mu_C^M]^{-1}\nu_e$, respectively. Assume there is a function $\var$ such that $\nabla\psi_C = \var\nabla\psi_0$. 
Then the second variation of $\iota^*h_C^M$ at $\nu_e$ is given by
\begin{equation}
    \label{e:SV2D}
    D^2(\iota^*h_C^M)(\nu_e)(\delta\nu,\delta\nu)
    = 
    \vv<\delta\nu, [\mu_C^M]^{-1}\delta\nu> 
    +
    \int_M \var\frac{\nabla\psi_0}{\nabla\Delta\psi_0}(\delta \om)^2\,dxdy
\end{equation}
where $\delta\nu = \ad(v)^*\nu_e\in T_{\nu_e}\mathcal{O}$,
$\delta\om = L_v\om_e$  
with $v\in\du$ and $\om_e = \om^{\nu_e}$.
\end{proposition}

\begin{proof}
The proof follows \cite[II.~Thm.~4.1]{AK98}.
For vector fields $v,u\in\du$ we may write the bracket as 
$[v,u] = -\nabla_v u + \nabla_uv = -\nabla^s(v\times_3 u)$ where $v\times_3 u$ is the third component of the cross product of the vectors $v,u$ when extended to $\R^3$. With $\delta\nu = L_v\nu_e$ 
and $u_C^e = [\mu_C^M]^{-1}(\nu_e)$ we have 
\begin{align*}
    \vv<\delta\nu,[v,u^e_C]>
    &= -\vv<\delta\nu,\nabla^s(v\times_3 u_C^e)>
    = \int_M (*d L_v\nu_e) \cdot (v\times_3 u_C^e) \,dxdy\\
    &= \int_M (L_v\om_e) \cdot (v\times_3 u_C^e) \,dxdy
    = \int_M \delta\om \cdot (v\times_3 u_C^e) \,dxdy
\end{align*}
Observe that 
$
 \om_e = *d\Pi_e = *d\mu_0^M\nabla^s\psi_0 = \Delta\psi_0
$
and that $\nabla\psi_0$ is parallel to $\nabla\Delta\psi_0$, since $\nu_e$ is an equilibrium of \eqref{e:E0}. 
Therefore, the equations 
$
 v\times_3 u_C^e 
 = \vv<v,\nabla\psi_C> 
 = \vv<v,\var\nabla\psi_0> 
$
and 
$
 \delta\om 
 = L_v\om_e
 = \vv<v,\nabla\Delta\psi_0>
$
imply that 
\[
 v\times_3 u_C^e = \var\frac{\nabla\psi_0}{\nabla\Delta\psi_0}\delta\om
\]
whence the result follows from Proposition~\ref{prop:SV}.
\end{proof}

We say that $\nu_e$ is formally stable for \eqref{e:lpa} if there is $\eps>0$, such that either
\begin{equation}
    \label{e:StabCrit_2D}
      D^2(\iota^*h_C^M)(\nu_e)(\delta\nu,\delta\nu)
      > \eps \int_M(\delta\om)^2\,dxdy
\end{equation}
or
\begin{equation}
    \label{e:StabCrit_2D_negdef}
    -D^2(\iota^*h_C^M)(\nu_e)(\delta\nu,\delta\nu)
    > \eps \int_M(\delta\om)^2\,dxdy.
\end{equation}

\section{Feedback stabilization of shear flow in a magnetic field}\label{sec:shear}

\subsection{Shear flow}
Consider the example of \cite[Ch.~II,~Ex.~4.6]{AK98}. 
Let $X,Y>0$, $M = \{(x,y): 0\le x\le X\pi, 0\le y\le Y\pi\}$, 
$\mu_0^M=\vv<.,.>$ the Euclidean metric 
and consider $\Pi_e = \Gamma(y)\,dx$. Then $\nu_e = [\Pi_e]\in\du^*$ is an equilibrium solution of the Euler equation~\eqref{e:E0} in $M$.  
In the following, we assume periodic boundary conditions in $x$, with periodicity $X\pi$. 

Let $\Gamma(y) = \sin(y+\by{\pi}{2})$ and assume that $Y>\by{1}{2}$. Then the flow has an inflection point at $y = \by{\pi}{2}$.  
The stream function is $\psi_0(x,y) = \psi_0(y) = \cos(y+\by{\pi}{2})$ and 
\begin{equation}
    \label{e:shear_factor}
    \frac{\nabla\psi_0}{\nabla\Delta\psi_0} = -1  
\end{equation}
whence criterion~\eqref{e:StabCrit_2D} cannot be satisfied. However,
if $X$ is sufficiently small then \eqref{e:StabCrit_2D_negdef} yields stability of $\nu_e$. The vorticity function of $\nu_e$ is $\om_e(y) = \Delta\psi_0(y) = -\cos(y+\by{\pi}{2})$.

\subsection{Magnetic field}
Let $K=S^1$, $\I_0=1$, $P = M\times K \to M$ a trivial principle bundle and consider the connection form $A_0: TM\to\ko = \R$ given by 
\[
 A_0(x,y) = a_0(y)\,dx.
\]
The associated magnetic field is $\mathcal{M} = -a_0'(y)\,dx\wedge dy$, which is a field that is orthogonal to the plane. We shall identify $\ko = \R = \ko^*$ and accordingly $\gau = \mathcal{F}(M,\R) = \gau^*$ from now on. 
The potential will be specified subject to certain conditions which are found below. In line with Remark~\ref{rem:CfromA}, finding a stabilizing control $C$ amounts to finding a suitable potential $A_0$. 

\subsection{The free system}\label{sec:free_magn_shear}
The free system corresponding to shear flow in a magnetic field $\mathcal{M}=dA_0$ is, as in \eqref{e:lp1}, \eqref{e:lp2}, \eqref{e:lp3}, 
  \begin{align}
    \dot{\nu} 
    &= -\ad(u)^*\nu - X\diamond q 
    \label{e:lp1_shear}
    \\
    \dot{q}
    &= -\rho^u(q) 
    \label{e:lp2_shear}\\
     \left(\begin{matrix} 
    u \\ 
    X
    \end{matrix}\right)
    &=   [\mu_0^P]^{-1}
    \left(\begin{matrix} 
    \nu \\ 
    q 
    \end{matrix}\right)
      \label{e:lp3_shear}
\end{align}  
where $\nu = [\Pi]\in\du^*$ and the metric $\mu_0^P$ is given by Lemma~\ref{pr:metric1} as
\[
 \mu_0^P
 = \mu^{KK}(\mu_0^M,\I_0,A_0)
 =  \left(\begin{matrix} 
    \mu_0^M + A_0^* A_0 &  A_0^* \\ 
    A_0 & 1
    \end{matrix}\right). 
\]
The variable $q$ corresponds to the charge of the particles and \eqref{e:lp2_shear} says that charge is conserved along flow lines.

\subsection{Control}
We wish to apply a force \eqref{e:TCforce} to the charge $q$, thereby obtaining a feedback control $q=-C\nu$ which yields closed loop equations that are of Lie-Poisson type \eqref{e:lpa}. 
Thus we use \eqref{e:Cgamma} to define the control as
\begin{equation}
    \label{e:C_shear}
    C 
    = \gamma R^{-1} A_0 [\mu_0^M]^{-1}
    = \gamma(1-\gamma A_0 [\mu_0^M]^{-1}[A_0^*])^{-1}A_0 [\mu_0^M]^{-1}.
\end{equation}

\begin{lemma}
Let $\tilde{\mu}_C^M := \Phi_{\gamma}\,dx\otimes dx + dy\otimes dy$ where $\Phi_{\gamma} := 1 - \gamma a_0^2$ and assume $\gamma a_0^2<1$. Then
\begin{equation}
    \label{e:mu_C^M_Shear}
    [\mu_C^M] 
    = [\tilde{\mu}_C^M]
\end{equation}
where $[\mu_C^M]$ is defined by \eqref{e:mu_C^M}.
\end{lemma}

\begin{proof}
Let $u = (u_1,u_2)^{\top}\in\du$. Now,
\begin{align*}
 (1+[A_0^*]C)[\tilde{\mu}_C^M](u)
 &= (1+[A_0^*]C)[\Phi_{\gamma}u_1dx + u_2 dy] \\
 &= (1 + \gamma[A_0^*]R^{-1}A_0[\mu_0^M]^{-1})[\Phi_{\gamma}u_1dx + u_2 dy]\\
 &= [\Phi_{\gamma}u_1dx + u_2 dy] 
    + \gamma[A_0^*]R^{-1}A_0\Big(
        \left(\begin{matrix}
            \Phi_{\gamma}u_1\\
            u_2
        \end{matrix}\right)
        - \nabla g_1 \Big) \\
 &= [\Phi_{\gamma}u_1dx + u_2 dy]
    + \gamma [A_0^*]R^{-1}\Big(a_0\Phi_{\gamma}u_1 - a_0\del_x g_1\Big) \\
 &= [\Phi_{\gamma}u_1dx + u_2 dy]
    + \gamma [a_0 R^{-1}\Big(a_0\Phi_{\gamma}u_1 - a_0\del_x g_1\Big) dx ]
\end{align*}
where $g_1$ is determined by $\Delta g_1 = \Phi_{\gamma}\del_x u_1 + \del_y u_2 = -\gamma a_0^2\del_x u_1$, since $a_0$ does not depend on $x$, together with Neumann boundary conditions. The boundary of $M$ is $\del M = \set{y = 0}\cup\set{y=Y\pi}$ because we assume periodicity in $x$. Therefore, $u\in\du$ implies that $u_2|\del M = 0$ whence the relevant boundary condition is $\vv<\nabla g_1|\del M, n> = u_2|\del M = 0$ where $n$ is the outward pointing unit normal vector. 

To show that  $[\mu_C^M]^{-1}[\tilde{\mu}_C^M](u) = [\mu_0^M]^{-1}(1+[A_0^*]C)[\tilde{\mu}_C^M](u) = u$ it thus suffices to check that
\[
 \gamma a_0^2 u_1 - \gamma a_0 R^{-1}\Big(a_0\Phi_{\gamma}u_1 - a_0\del_x g_1\Big) = 0.
\]
We have
\begin{align}
\label{e:g_1g_2}
    R(a_0u_1)
    &= (1-\gamma A_0[\mu_0^M]^{-1}[A_0^*])(a_0 u_1)
    = a_0u_1 -\gamma A_0\Big(
         \left(\begin{matrix}
            a_0^2u_1\\
            0
        \end{matrix}\right)
        - \nabla g_2 
    \Big)
\end{align}
where $g_2$ is determined by $\Delta g_2 = a_0^2\del_x u_1 = -\gamma^{-1}\Delta g_1$, since $a_0$ does not depend on $x$, together with Neumann boundary conditions. The relevant boundary condition is $\vv<\nabla g_2|\del M, n> = 0$. Therefore, $g_2 = -\gamma^{-1}g_1$ and 
\begin{align*}
 \eqref{e:g_1g_2}
 &= a_0u^1 - \gamma a_0(a_0^2u_1 + \gamma^{-1}\del_x g_1)
 = a_0\Phi_{\gamma}u_1 - a_0\del_x g_1.
\end{align*}
\end{proof}

Therefore, in this case, $[\mu_C^M]$ does come from a Riemannian metric $\tilde{\mu}_C^M$ on $M$, and we shall omit the tilde from now on and write $\mu_C^M = \tilde{\mu}_C^M$. Compare with Remark~\ref{rem:mu_C}.

In order for $\mu_C^M$ to be a well-defined metric it is required that $\gamma$ and $a_0^2$ have to be chosen such that 
\begin{equation}
    \label{e:shear_cond1}
    \textup{Condition: }\qquad 
    \gamma a_0^2 < 1
\end{equation}
or else \eqref{e:mu_C^M_Shear} would not be positive definite. 

The map $T: \gau^*\to\gau^*$ is defined according to \eqref{e:DefT} and it can be checked that \eqref{e:g_1g_2} implies
\begin{equation}
    \label{e:T_shear}
    T 
    = 1 + \gamma + \gamma(1+\gamma)R^{-1}A_0[\mu_0^M]^{-1}[A_0^*]
    = \frac{1+\gamma}{1-\gamma a_0^2}
\end{equation}
which is invertible under assumptions \eqref{e:shear_cond1} and \eqref{e:a_0NoInfl}. Similarly, \eqref{e:g_1g_2} yields
\begin{equation}
    \label{e:C_naive_shear}
    C = \frac{\gamma}{1-\gamma a_0^2}A_0[\mu_0^M]^{-1}.
\end{equation}

\subsection{Formal stability condition}\label{sec:shear_form_stab_cond}
Let $v\in\du$. With $k(y) := \int_0^y \Gamma(z)\del_z(\Phi_{\gamma}^{-1}\Gamma)(z)\,dz$ it follows that $\vv<\ad([\mu_C^M]^{-1}\nu_e)^*\nu_e, v> = \int_M\vv<v, \nabla k>\,dxdy  = 0$, whence
\[
 \ad([\mu_C^M]^{-1}\nu_e)^*\nu_e = 0
\]
and $\nu_e = [\Pi_e]$ is also an equilibrium of the controlled Lie-Poisson equation \eqref{e:lpa}. Further, we have 
\begin{equation}
    \label{e:stream_shear}
    \nabla^s\psi_C^e 
    = [\mu_C^M]^{-1}\nu_e 
    = \Phi_{\gamma}^{-1}[\mu_0^M]^{-1}\nu_e 
    = \Phi_{\gamma}^{-1}\nabla^s\psi_0^e. 
\end{equation}
Therefore, the second variation of $\iota^*h_C^M$ at $\nu_e$, where $\iota$ is the inclusion of the coadjoint orbit through $\nu_e$, is given by Proposition~\ref{prop:SV2D}.  
We want to find conditions on $\gamma$ and $a_0$ so that  \eqref{e:SV} is negative definite. (Since $\nabla\psi_0/\nabla\Delta\psi_0 < 0$ and $\Phi_{\gamma}^{-1}>0$, it cannot be positive definite.)

To this end, let $\nu = [\mu_C^M]\nabla^s \psi_C$ and observe
\begin{align}
    \int_M \vv< \nu, [\mu_C^M]^{-1}\nu > \,dxdy
    &= \int_M \left<
        \left(
        \begin{matrix}
            \Phi_{\gamma} &  \\ 
              & 1 
        \end{matrix}
        \right)
        \left(
        \begin{matrix}
            -(\psi_C)_y \\ 
            (\psi_C)_x
        \end{matrix}
        \right)
        ,
                \left(
        \begin{matrix}
            -(\psi_C)_y \\ 
            (\psi_C)_x
        \end{matrix}
        \right)
        \right> \,dxdy 
        \notag \\
    &= 
    - \int_M \psi_C\textup{div}
    \left(
    \left(
    \begin{matrix}
    1 &  \\
      & \Phi_{\gamma}
    \end{matrix}
    \right)
    \nabla\psi_C
    \right) 
    \,dxdy 
    \label{e:Ekin1}
\end{align}
where 
\[
 \Delta^C\var  
 := \textup{div}
    \left(
    \left(
    \begin{matrix}
    1 &  \\
      & \Phi_{\gamma}
    \end{matrix}
    \right)
    \nabla \var  
    \right) 
\]
is an elliptic operator in divergence form. To transform this to a drifted Laplacian, consider the change of variable
\[
 z(y) = \int_0^y\frac{1}{\sqrt{\Phi_{\gamma}(s)}}\,ds
\]
which is well-defined and invertible since the integrand is strictly positive. Thus $y(z)$ exists and satisfies $y'(z) = (z'(y(z)))^{-1} = \sqrt{\Phi_{\gamma}(y(z))}$ and $y''(z) = \by{1}{2}\Phi_{\gamma}'(y(z))$. 
For $\var = \var(x,y)$  let $\tilde{\var}(x,z) = \var(x,y(z))$. It follows that 
\begin{align*}
    \Delta\tilde{\var}(x,z)
    &= \del_x^2\var(x,y(z)) 
        + y''(z)\del_y\var(x,y(z)) 
        + (y'(z))^2\del_y^2\var(x,y(z))\\
    &= \Delta^C\var(x,y(z)) 
        - \by{1}{2}\Phi_{\gamma}'(y(z))\del_y\var(x,y(z))
\end{align*}
whence we obtain the drifted Laplacian, $\Delta_g = \Delta - \nabla g$,
\begin{equation}
    \label{e:shear_Delta_transf}
    \Delta^C\var (x,y(z) )
    = 
    \Big(
        \Delta 
        + \by{1}{2}\vv<\nabla(\log\tilde{\Phi}_{\gamma},\nabla>
    \Big)\tilde{\var} (x,z)  
    = \Delta_g\tilde{\var} (x,z) 
\end{equation}
where $\tilde{\Phi}_{\gamma}(z) = \Phi_{\gamma}(y(z))$ and $g := -\by{1}{2}\log\tilde{\Phi}_{\gamma}$. 
We apply Theorem~\ref{thm:FLL} to $\Delta_g$. Therefore, we look for conditions on $\gamma$ and $a_0$ yielding an (ideally) large constant $K\in\R$ such that 
\begin{equation}
    \label{e:Hess}
    \textup{Hess}_g \ge K \mu_0^M
\end{equation}
where $\textup{Hess}_g$ is the Hessian matrix of $g$. The only non-zero entry in $\textup{Hess}_g$ is $\del_z^2 g$. Thus $K=0$ is the largest possible constant. Since
\begin{equation}
  \label{e:Hess_g}
  \del_z^2g(z) 
  = 
  -\frac{1}{2}\frac{\Phi_{\gamma}\Phi_{\gamma}''-\frac{1}{2}(\Phi_{\gamma}')^2}{\Phi_{\gamma}}(y(z))
  =
  \frac{1}{2}\frac{2\gamma((a_0')^2+a_0a_0'')\Phi_{\gamma}+\frac{1}{2}(\Phi_{\gamma}')^2}{\Phi_{\gamma}}(y(z))
\end{equation}
and $\Phi_{\gamma}>0$, a sufficient condition for $K=0$ is that $\gamma$ and $a_0a_0''$ are positive. We shall assume from now on 
\begin{equation}
    \label{e:a_0NoInfl}
    \textup{Condition: }\qquad
    \gamma>0
    \textup{ and }
    a_0a_0'' \ge 0.
\end{equation}
Theorem~\ref{thm:FLL} now implies that the first Dirichlet eigenvalue $\lam_1(\gamma)$ of $\Delta_g$ on 
$M_{\gamma} = [0,X\pi]\times[0,Z_{\gamma}]$ satisfies
\begin{equation}
    \label{e:shear_lam1}
    \lam_1(\gamma) \ge \frac{\pi^2}{\pi^2 X^2 + Z_{\gamma}^2}
\end{equation}
where $Z_{\gamma}$ is 
\begin{equation}
    \label{e:Zgamma}
    Z_{\gamma} 
    = z(Y\pi)
    = \int_0^{Y\pi}\frac{1}{\sqrt{\Phi_{\gamma}(s)}}\,ds.
\end{equation}
The change of variable $y = y(z)$ and the reverse Poincare inequality \eqref{e:RevPoincare} 
together with $y'(z) = e^{-g(z)}$ therefore yield
\begin{align}
    \notag
    \eqref{e:Ekin1}
    &= 
    -\int_0^{X\pi}\int_0^{Y\pi} \psi_C(x,y)\Delta^C\psi_C(x,y)\,dxdy \\ 
    \notag
    &=    
    -\int_0^{X\pi}\int_0^{Z_{\gamma}} \tilde{\psi}_C(x,z)\Delta_g\tilde{\psi}_C(x,z)\,y'(z)\,dxdz \\
    \notag
    &\le 
    \lam_1(\gamma)^{-1}
    \int_0^{X\pi}\int_0^{Z_{\gamma}} \Big(\Delta_g\tilde{\psi}_C(x,z)\Big)^2\,y'(z)\,dxdz \\
    &= 
    \lam_1(\gamma)^{-1}
    \int_0^{X\pi}\int_0^{Y\pi} \Big(\Delta^C\psi_C(x,y)\Big)^2\,dxdy 
    \label{e:Ekin2}
\end{align}
where, as in \eqref{e:shear_Delta_transf}, $\tilde{\psi}_C(x,z) = \psi_C(x,y(z))$. 
Since
$\nabla^s\psi_C = (-(\psi_C)_y,(\psi_C)_x)^{\top} = [\mu_C^M]^{-1}\nu$ and 
\[
 \om 
 = *d\nu
 = *d \mu_C^M \nabla^s \psi_C
 = *d(-\Phi_{\gamma} (\psi_C)_y dx + (\psi_C)_x dy)
 = \Delta_C \psi_C 
\]
it follows that 
\begin{align}
\notag
    \eqref{e:Ekin2} 
    &= \lam_1(\gamma)^{-1}
        \int_0^{X\pi}\int_0^{Y\pi} \om^2\,dxdy 
    \le
    \frac{\pi^2 X^2 + Z_{\gamma}^2}{\pi^2} \int_M \om^2\,dxdy.
    \notag
\end{align}
With \eqref{e:SV2D} we can therefore express the second variation of $h_C^M$, restricted to the orbit $\mathcal{O}$ through $\nu_e$, as
\begin{align*}
    D^2(\iota^*h_C^M)(\nu_e)(\nu,\nu)
    &= 
    \vv<\nu, [\mu_C^M]^{-1}\nu> 
    +
    \int_M \Phi_{\gamma}^{-1}(y)\frac{\nabla\psi_0}{\nabla\Delta\psi_0} \om^2\,dxdy \\
    &\le
    \frac{\pi^2 X^2 + Z_{\gamma}^2}{\pi^2} \int_M \om^2\,dxdy
    - \int_M \Phi_{\gamma}^{-1}(y) \om^2\,dxdy.
\end{align*}
Let $\overline{\Phi}_{\gamma} = \textup{max}_{0\le y \le Y\pi}\Phi_{\gamma}(y)$ and $\underline{\Phi}_{\gamma} = \textup{min}_{0\le y \le Y\pi}\Phi_{\gamma}(y)$. 
Notice that $Z_{\gamma}^2\le \pi^2 Y^2 / \underline{\Phi}_{\gamma}$
and
$\int_M\Phi_{\gamma}^{-1}\om^2\,dxdy 
\ge 
\overline{\Phi}_{\gamma}^{-1}\int_M\om^2\,dxdy$.
It follows that a sufficient condition for the second variation  
$D^2(\iota^*h_C^M)(\nu_e)(\nu,\nu)$ to be negative definite is
\begin{equation}
    \label{e:shear_ND_cond}
    \textup{Condition: }\qquad
    \overline{\Phi}_{\gamma} X^2 + \frac{\overline{\Phi}_{\gamma}}{\underline{\Phi}_{\gamma}}Y^2 < 1.
\end{equation}

\subsection{Stability of controlled shear flow}
We recall the main assumptions of this section.
Let $\by{1}{2}\le Y<1$ and $X>0$ arbitrary. Consider shear flow on $M = [0,X\pi]\times[0, Y\pi]$ and $\Pi_e = \sin(y+\by{\pi}{2})dx$, as above. 
Then $\nu_e = [\Pi_e]$ is a stationary solution of the (uncontrolled) Euler equation
\eqref{e:lp1_shear}, which is stable for sufficiently small $X$ and unstable for large $X$. 
The fluid is assumed to consist of charged particles subject to an external magnetic field $\mathcal{M} = dA_0$ where $A_0 = a_0(y)dx$. 


In the following, stability is understood in the nonlinear sense and with respect to perturbations whose circulations around $\del M$ vanish. 

\begin{theorem}
\label{thm:shear_formal_stab}
Let $a_0(y) = b(\om_e(y))$, where $\om_e = -\cos(y+\by{\pi}{2})$ is the vorticity function associated to $\Pi_e$, 
let
$b: [0,1]\to[\underline{b},\overline{b}]\subset(0,1)$, $\om\mapsto \overline{b} - (\overline{b}-\underline{b})\om$ and 
\begin{align*} 
 &\overline{b} = \sqrt{1-\by{\alpha}{\beta}},
 \quad
 \underline{b} = \sqrt{1-\alpha},
 \quad
 \alpha = \frac{r}{X^2},\quad
 \beta = \frac{Y^2+r}{Y^2},
 \quad 
 r = \frac{1-Y^2}{3}.
\end{align*}
Then the control \eqref{e:C_shear}, with $\gamma = 1$, yields a feedback system \eqref{e:lpa} which is Lie-Poisson, and $\nu_e$ is a stable equilibrium for this system. 
\end{theorem}

Hence the control \eqref{e:C_shear} stabilizes the equilibrium $\nu_e$ for $\gamma=1$.
The magnetic field is 
\[
\mathcal{M} 
= -a_0'\,dx\wedge dy 
= (\overline{b}-\underline{b})\om_e'\,dx\wedge dy 
= (\overline{b}-\underline{b})\sin(y+\by{\pi}{2})\,dx\wedge dy
= (\overline{b}-\underline{b})\Pi_e \wedge dy
\]
and the stabilizing control $C: \du^*\to\gau^*$ is, according to equation~\eqref{e:C_naive_shear},
\[
 C = \frac{b(\om_e)}{1-b(\om_e)^2}dx\circ[\mu_0^M]^{-1}.
\]
Concretely, $\nu_e$ is stabilized, with respect to perturbations in $\nu$, by applying the closed loop equations that arise when feeding the control 
\begin{equation}
    \label{e:fb_cntrl}
     q = -\frac{b(\om_e)\,dx([\mu_0^M]^{-1}\nu)}{1-b(\om_e)^2}
\end{equation}
into \eqref{e:lp1_shear} and \eqref{e:lp3_shear}. 

\begin{remark}[Physical significance of the control]
The unforced shear flow is given by equations \eqref{e:lp1_shear}-\eqref{e:lp3_shear}. As shown in Appendix~\ref{app:GM} this system corresponds to the charged Euler equations \eqref{1e:EM_u_d}-\eqref{1e:cont_d}. 
The above control now means that the charge $q$ of the physical system is subjected to the feedback law \eqref{e:fb_cntrl} which depends on observations of the fluid's momentum $\nu$. 
\end{remark}

\begin{proof}
The conditions for formal stability are, by construction,  \eqref{e:shear_cond1}, \eqref{e:a_0NoInfl} and \eqref{e:shear_ND_cond}.
To see that these are satisfied, note that 
\begin{equation}
\label{e:noInflcos}
     a_0(y)a_0''(y)
 = - b(\om_e(y))(\overline{b}-\underline{b})\om_e''(y)
 = - b(\om_e(y))(\overline{b}-\underline{b})\cos(y+\by{\pi}{2})\ge0
\end{equation}
for $y\in[0,Y\pi]$. Further, we have $a_0<1$ and 
\begin{equation}
    \label{e:form_sta_expl}
     \overline{\Phi}_1X^2+\frac{\overline{\Phi}_1}{\underline{\Phi}_1}Y^2
 = (1-\underline{b}^2) X^2 + \frac{1-\underline{b}^2}{1-\overline{b}^2} Y^2 
 = \alpha X^2 + \beta Y^2 
 = 1 - r 
 < 1. 
\end{equation}
Consider the functional relation $\Psi_0\circ\om_e = \psi_0^e$ which is simply given by $\Psi_0(\om) = -\om$.
Here  $\psi_0^e(y) = \cos(y+\by{\pi}{2})$ is the stream function associated to $[\mu_0^M]^{-1}\nu_e$.
To find a similar relation for $\psi_C^e$, which is the stream function of $[\mu_C^M]^{-1}\nu_e$, use \eqref{e:stream_shear} and consider
\begin{align*}
    \psi_C^e(y)
    &= \int_0^y \del_s\psi_C^e(s)\,ds
    = \int_0^y \Phi_1(s)^{-1}\del_s\psi_0^e(s)\,ds
    = \int_0^y (1-b(\om_e(s))^2)^{-1}\del_s(\Psi_0\circ\om_e)(s)\,ds\\
    &= \int_0^y (1-b(\om_e(s))^2)^{-1} \Psi_0'(\om_e(s)) \om'_e(s) \,ds
    = \int_0^y \del_s(\Psi_C\circ\om_e)(s)\,ds
    = \Psi_C(\om_e(y))
\end{align*}
where $\Psi_C(\om) := \int_0^{\om} (1-b(\eta)^2)^{-1} \Psi_0'(\eta) \,d\eta = -\int_0^{\om} (1-b(\eta)^2)^{-1} \,d\eta$. We thus have $\nabla\psi_C^e = \Psi_C'(\om_e)\nabla\om_e$.

The rest of the proof follows now \cite[Ch.~2]{AK98}.
Extend $b$ to a smooth function on $\R$ such that $0 < \underline{b}-\eps \le b(\om) \le \overline{b}+\eps < 1$ for a small $\eps>0$ and all $\om\in\R$. Thus $\Phi_1$ and $\Psi_C$ are extended to $\R$ as-well.
For $\tau\in\R$ let
\[
 \phi_C(\tau) = \int_0^{\tau}\Psi_C(\theta)\,d\theta
\]
whence 
\[
 -\phi_C''(\om) 
 = -\Psi_C'(\om) 
 = \frac{1}{\Phi_1(\om)}
 \ge \frac{1}{\overline{\Phi}_1+2\underline{b}\eps-\eps^2}
 = \frac{1}{\overline{\Phi}_1} - \kappa 
\]
for all $\om\in\R$ and where $\kappa = \eps(2\underline{b}-\eps)/(\overline{\Phi}_1(\overline{\Phi}_1 + (2\underline{b}-\eps)\eps ) )$.

For $\nu\in\du^*$ define
\[
 H(\nu) = h_C^M(\nu) + \int_M\phi_C\circ\om\,dxdy
\]
where $\om$ is the vorticity of $\nu$. Since \eqref{e:lpa} is vorticity preserving, $H$ is conserved along solutions. Hence the same is true for $\hat{H}(\nu) = H(\nu+\nu_e)-H(\nu_e)$,
where $\nu + \nu_e$ is a perturbed solution. 
One decomposes $\hat{H} = \hat{H}_1 + \hat{H}_2$ where
\[
 \hat{H}_1(\nu)
 = \int_B\Big(\vv<\nu,[\mu_C^M]^{-1}\nu_e> + \phi_C'(\om_e)\om\Big)\,dxdy 
\]
and 
\[
  \hat{H}_2(\nu)
  = \int_B\Big(\by{1}{2}\vv<\nu,[\mu_C^M]^{-1}\nu> 
  + \phi_C(\om+\om_e) - \phi_C(\om_e)
  - \phi_C'(\om_e)\om\Big)\,dxdy .
\]
We only consider perturbations $\nu$ whose circulation around $\del M$ vanishes. The circulation of $\nu+\nu_e$ around $\del M$ is preserved by the Kelvin Circulation Theorem and equals the circulation of $\nu_e$ around $\del M$. It follows that  $\hat{H}_1(\nu_t)=0$, and $\hat{H}_2(\nu_t)$ is constant in $t$. (The details of this argument are in the proof of \cite[Ch.~2, Thm.~4.3]{AK98}.) 

The formal stability condition~\eqref{e:form_sta_expl} now implies that we can bound the perturbation's vorticity by
\begin{align*}
 -\hat{H}_2(\nu_0)
 &= -\hat{H}_2(\nu_t)
 = 
 -\int_M\Big(
 \by{1}{2}\vv<\nu_t,[\mu_C^M]^{-1}\nu_t> 
 + \phi_C(\om_t+\om_e) - \phi_C(\om_e) - \phi_C'(\om_e)\om_t
 \Big)\,dxdy  \\
 &\ge 
 -\by{1}{2}(X^2+\underline{\Phi}_1^{-1} Y^2)\int_M\om_t^2\,dxdy
 +\by{1}{2}(\overline{\Phi}_1^{-1}-\kappa)\int_M\om_t^2\,dxdy \\
 &= 
 \by{1}{2}\overline{\Phi}_1^{-1}
 \Big(
    -\overline{\Phi}_1 X^2 
    - \overline{\Phi}_1\underline{\Phi}_1^{-1}Y^2
    + 1 - \kappa \overline{\Phi}_1
 \Big)\int_M\om_t^2\,dxdy\\
 &= 
 \by{1}{2}\overline{\Phi}_1^{-1}
 \Big(r - \kappa\overline{\Phi}_1 \Big)\int_M\om_t^2\,dxdy
\end{align*}
and this completes the proof by choosing $\eps$ such that $\eps(2\underline{b}-\eps)/(\overline{\Phi}_1 + (2\underline{b}-\eps)\eps ) < r$.
\end{proof}

\begin{remark}
We have assumed that $Y\ge\by{1}{2}$ to ensure that $\om_e$ assumes all values between $0$ and $1$, otherwise the definition of $b(\om)$ would have to be slightly adapted as $b(\om) = \overline{b}-(\overline{b}-\underline{b})\om/\textup{max}(\om_e)$. The assumption $Y<1$ is more restrictive as it is necessary to apply the bound of Theorem~\ref{thm:FLL}. If there is a sharper bound for the first eigenvalue $\lam_1$ of the drifted Laplacian $\Delta_g$ defined in \eqref{e:shear_Delta_transf}, then this assumption may be relaxed. However, the range of $y$ is also constrained by the requirement \eqref{e:noInflcos}.  
\end{remark}


\section{Appendix: Drifted Laplacian}\label{sec:appendix}

\subsection{Drifted Laplacian}\label{app:drift_Delta}
Let $M$ be a compact domain in $\R^n$ with smooth boundary $\del M$. 
Let $g: M\to\R$ be a smooth function and define the drifted Laplacian 
\begin{equation}
    \label{e:DriftDelta}
    \Delta_g = \Delta - \vv<\nabla g, \nabla>
\end{equation}
where $\Delta$ is the ordinary Laplacian. 
One can check that $\Delta_g$ is symmetric with respect to the weighted measure $e^{-g}dx$:
\begin{equation}
    \label{e:Green}
    -\int_M\var(\Delta_g\psi)\,e^{-g} dx
    = \int_M\vv<\nabla\var,\nabla\psi>\,e^{-g} dx
    = -\int_M\psi(\Delta_g\var)\,e^{-g} dx
\end{equation}
where $\var,\psi$ are functions on $M$ which vanish at the boundary. The drifted Laplacian is also called Witten-Laplacian in \cite{FLL13}. The eigenvalue equation for $\Delta_g$ is $-\Delta_g\psi = \lam\psi$. 

\begin{theorem}[\cite{FLL13}]\label{thm:FLL}
Let $(M,\mu_0^M)$ be an $n$-dimensional compact Riemannian manifold, and let $g\in C^2(M)$. Suppose that there exists a constant $K\in\R$ such that
\begin{equation}
    \label{e:FLL_K}
     \textup{Ric}+\textup{Hess}_g \ge K\mu_0^M.
\end{equation}
Then the first non-zero eigenvalue $\lam_1$ of the Witten-Laplacian $\Delta_g$ satisfies
\begin{equation}
    \label{e:FLL_lam1}
    \lam_1 \ge \sup_{s\in(0,1)}\set{4s(1-s)\frac{\pi^2}{d^2} + sK}
\end{equation}
where $d$ is the diameter of $(M,\mu_0^M)$.
\end{theorem}

In Section~\ref{sec:shear_form_stab_cond} we apply this result to the case of a flat $2$-dimensional domain to bound the constant $\lam_1$ in the reverse Poincar\'e inequality
\begin{equation}
    \label{e:RevPoincare}
    \int_M(\Delta_g\var)^2\,e^{-g}dx
    \ge -\lam_1\int_M \var(\Delta_g\var)\,e^{-g}dx.
\end{equation}

\section{Appendix: some concepts in geometric mechanics}\label{app:GM}

This appendix collects the some notation and concepts from geometric mechanics as they are used in the body of the paper. 

\subsection{Kaluza-Klein metrics}
Let $\pi: P\to M$ be a finite dimensional principal fiber bundle with structure group $K$ acting from the right. The principal right action is denoted by $r^k: P\to P$ for $k\in K$. Let $\ko$ be the Lie algebra of $K$ and $\A\in\Om^1(M,\ko)$ be a principal bundle connection form. 
For $k\in K$ and $\xi\in TP$ it follows that $\A(Tr^k.\xi) = \Ad(k)^{-1}\A(\xi)$. 
The vertical space of the principal bundle is canonically defined as $\Ver = \ker T\pi$. The horizontal space corresponding to the connection is defined as $\Hor = \ker\A$.
See \cite{MichorDG} for background on principal bundles. 

Assume that $K$ is compact and endowed with a positive definite symmetric bilinear form $\I$ and that $M$ carries a Riemannian metric $\mu^M$. Then the data $(\mu^M,\I,\A)$ give rise to a Riemannian metric $\mu^P$ on $P$ by the following prescriptions:
\begin{itemize}
    \item 
    On $\Hor=\ker\A\cong P\times_M TM$ the metric $\mu^P$ is given by $\mu^M$. 
    \item
    On $\Ver\cong P\times\ko$ the metric $\mu^P$ is given by $\I$. 
    \item
    $\Hor$ and $\Ver$ are orthogonal to each other with respect to $\mu^P$.  
\end{itemize}
This metric is called the Kaluza-Klein metric. It can be expressed, with $\xi,\eta\in T_pP$, as
\begin{equation}
\label{e:app-KK}
    \mu_p^P(\xi,\eta)
    = \mu^M_{\pi(p)}(T_p\pi.\xi, T_p\pi.\eta)
        + \I(\A_p\xi,\A_p\eta). 
\end{equation}
To highlight the dependence of $\mu^P$ on the triple $(\mu^M,\I,\A)$ we write $\mu^P = \mu^{KK}(\mu^M,\I,\A)$. See also \cite{MR-mas}.

\subsection{Trivial principle bundle and metric formula}
Let the notation be as in Section~\ref{sec:SD_struc}.    

\begin{lemma}[Metric formula]\label{pr:metric1}
Assume that $TM=M\times\R^n$ and that $P=M\times K$.
Let  $(\mu^M,\I,A)$ denote a set of Kaluza-Klein data. 
\begin{enumerate}[\up (1)]
    \item 
    Using the splitting 
$TP = M\times K\times\R^n\times\ko$, 
the isomorphism $\mu^P = \mu^{KK}(\mu^M,\I,A): \R^n\times\ko\to\R^n\times\ko^*$ can be expressed as 
\begin{equation}
    \mu^P 
    =
    \left(\begin{matrix}
    \mu^M+A^*\I A && A^*\I \\
    \I A && \I
    \end{matrix}\right)
\end{equation}
with inverse
\begin{equation}
    (\mu^P)^{-1} 
    =
    \left(\begin{matrix}
    (\mu^M)^{-1} && -(\mu^M)^{-1}A^* \\
    -A(\mu^M)^{-1} && \I^{-1}+A(\mu^M)^{-1}A^*
    \end{matrix}\right)
\end{equation}
\item
Since $\ao = \du\times\gau$, we obtain the isomorphisms
\begin{align*}
    &[\mu^P]: 
        \ao\to\ao^*,\quad
        \Big(u,X\Big)^{\top}
        \mapsto 
        \Big([ (\mu^M + A^*\I A)u + A^*\I X], \I(Au + X)\Big)^{\top} \\
    &[\mu^P]^{-1}:
        \ao^*\to\ao,\quad 
        \Big([\Pi],q\Big)^{\top} 
        \mapsto
        \Big( [\mu^M]^{-1}[\Pi-A^*q],
          -A[\mu^M]^{-1}[\Pi] + Q q
        \Big)^{\top}
\end{align*}
where $Q q := \I^{-1}q + A[\mu^M]^{-1}[A^*]q$ and $[A^*]q := [A^*q]$.
\item
Let $L_1: \du^*\to\du^*$ and $L_2: \gau^*\to\gau^*$ be invertible operators such that $L_1^* = [\mu^M]^{-1}L_1[\mu^M]$, $L_2^* = \I^{-1}L_2\I$ and $\Ad(k)^* L_2\I \Ad(k) = L_2\I$. Let $\tau: \du\to\gau$ be a linear map. Then the Kaluza-Klein inner product $[\mu] := \mu^{KK}(L_1[\mu^M],L_2\I, A+\tau)$ on $\ao$, given by
\[
 [\mu]\Big((u,X),(v,Y)\Big) 
 = \vv<L_1[\mu^M] u,v> + \vv< L_2\I (X + A_{\tau}u), Y + A_{\tau}(v) >
\]
for $(u,X),(v,Y)\in\du\times\gau$ and where $A_{\tau} = A+\tau$, can be expressed as
\begin{equation}
    [\mu] 
    =
    \left(\begin{matrix}
    L_1[\mu^M]+[A_{\tau}^*]L_2\I A_{\tau} && [A_{\tau}^*] L_2 \I \\
    L_2\I A_{\tau} && L_2\I
    \end{matrix}\right)
\end{equation}
with inverse
\begin{equation}
    [\mu]^{-1} 
    =
    \left(\begin{matrix}
    [\mu^M]^{-1}L_1^{-1} && -[\mu^M]^{-1}L_1^{-1}[A_{\tau}^*] \\
    -A_{\tau}[\mu^M]^{-1}L_1^{-1} && \I^{-1}L_2^{-1} + A_{\tau}[\mu^M]^{-1}L_1^{-1}[A_{\tau}^*]
    \end{matrix}\right)
\end{equation}
\end{enumerate}
\end{lemma}

\begin{proof}
Let $(u,X)^{\top},(v,Y)^{\top}\in\R^n\times\ko$. It suffices to verify that 
\begin{align*}
\left<
\mu^P
 \left(\begin{matrix}
    u\\
    X
    \end{matrix}\right),
 \left(\begin{matrix}
    v\\
    Y
    \end{matrix}\right)
\right>
&=
\vv<\mu^M(u),  v>
+ \vv<\I(X+Au), Y+Av> \\
&=
\left<
 \left(\begin{matrix}
    \mu^M + A^*\I A &&  A^*\I \\
    \I A && \I
    \end{matrix}\right)
 \left(\begin{matrix}
    u\\
    X
    \end{matrix}\right),
 \left(\begin{matrix}
    v\\
    Y
    \end{matrix}\right)
\right>
\end{align*}
\end{proof}

\subsection{Lie-Poisson systems}
Let $G$ be a finite or infinite-dimensional Lie group with dual Lie algebra $\gu^*$. On $\gu^*$ there exists a canonical Poisson bracket. This is defined in terms of the Lie bracket $[.,.]$ on $\gu$ as 
\begin{equation}
    \label{e:app-KKS}
    \{f,g\}
    = 
    -\vv<\mu,[\frac{\delta f}{\delta \mu},\frac{\delta g}{\delta\mu}]>
\end{equation}
for $\mu\in\gu^*$ and functions $f,g:\gu^*\to\R$. Here $\delta f/\delta \mu\in\gu$ is the variational derivative,
\[
 \vv<\nu, \frac{\delta f}{\delta \mu}>
 = \dd{t}{}|_0 f(\mu + t\nu). 
\]
The Hamiltonian vector field corresponding to a function $h: \gu^*\to\R$ is $X_h = \{h,.\}$. (This choice of sign for the Hamiltonian vector field is consistent with \cite{MichorDG} but should be reversed when compared to \cite{MR-mas}.)
The equations of motion associated to $h$ and \eqref{e:app-KKS} are therefore 
\begin{equation}
    \label{e:app-eom}
    \dot{f}
    = L_{X_h}f
    = \{h,f\}
    = -\vv<\ad(\frac{\delta h}{\delta\mu})^*\mu,\frac{\delta f}{\delta\mu}>  
\end{equation}
where $\ad(X)^*$ is the dual of the adjoint operation  $\ad(X).Y = [X,Y]$. The bracket \eqref{e:app-KKS} can be obtained by Poisson reduction of $T^*G$ with respect to the right action of $G$ on itself, see e.g.\ \cite[Example~34.15]{MichorDG}.

\subsection{Charged Euler equations as a Lie-Poisson system}
If $\aut_0^* = \diff_0^*\times\gau^*$ is the dual to $\aut_0 = \X_0(M)\times\F(M,\ko)$, as defined in Section~\ref{sec:SD_struc}, then \eqref{e:app-KKS} can be used to define a Poisson structure on $\aut_0^*$ even though $\A$ is only a topological group, and  not a Lie group.  See \cite{MRW84}. 
The relevant bracket follows now from the semi-direct product multiplication \eqref{e:s-d} in $\A$. For $(u,X)\in\aut_0=\X_0(M)\times\F(M,\ko)$ and $(v,Y)\in\aut_0=\X_0(M)\times\F(M,\ko)$ this is given by 
\[
 \left[
    \left(
    \begin{matrix}
    u\\ 
    X
    \end{matrix}
    \right)
    ,
    \left(
    \begin{matrix}
    v\\ 
    Y
    \end{matrix}
    \right)
 \right]
 =
    \ad    \left(
    \begin{matrix}
    u\\ 
    X
    \end{matrix}
    \right)
    \left(
    \begin{matrix}
    v\\ 
    Y
    \end{matrix}
    \right)
 =
     \left(
    \begin{matrix}
    [u,v]\\ 
    [X,Y]_{\mathfrak{k}} - L_u X + L_v Y
    \end{matrix}
    \right)
\]
where $[u,v] = -L_u v$ is minus the Lie bracket of vector fields and $[X,Y]_{\mathfrak{k}}$ is defined pointwise. The corresponding coadjoint operation follows from evaluating $(\nu,q)\in\aut_0^* = \Om^1(M)/d\F \times \F(M,\ko^*)$ on this bracket,
\begin{equation}
    \label{e:app-coad}
    \ad   
    \left(
    \begin{matrix}
    u\\ 
    X
    \end{matrix}
    \right)^*
    \left(
    \begin{matrix}
    \nu\\ 
    q
    \end{matrix}
    \right)
 =
    \left(
    \begin{matrix}
    \ad(u)^*\nu + X\diamond q\\ 
    \rho^uq + \ad_{\mathfrak{k}}^*(X)q
    \end{matrix}
    \right)
\end{equation}
where $\ad(u)^*\nu = L_u \nu = L_u[\Pi] = [L_u\Pi]$ is now the Lie derivative of the class $\nu = [\Pi]$ of the one-form $\Pi$ in $\Om^1(M)/d\F$. The diamond notation is defined by 
\[
 \vv<X\diamond q, v> = \int_M q L_v X \,\textup{vol}_M
\]
for $v\in\X_0(M)$. 
The Euclidean structure of $M$ implies that $X\diamond q = [q dX] \in \Om^1(M)/d\F$. Further, $\rho^u q = (\rho^{-u})^*q = \vv<u,\nabla q>$ has been introduced in Section~\ref{sec:SD_struc}. 

Using Lemma~\ref{pr:metric1}  the variational derivatives of the Kaluza-Klein Hamiltonian $h_0$ defined in \eqref{e:h_0} are 
\[
\left(
    \begin{matrix}
    \delta h_0/ \delta\nu\\ 
    \delta h_0/ \delta q
    \end{matrix}
    \right)
    =
\left(
    \begin{matrix}
    u\\ 
    X
    \end{matrix}
    \right)
    =
    \left(
    \begin{matrix}
    [\mu^M]^{-1}\nu - [A^*]q    \\ 
    -A u + \I^{-1}q
    \end{matrix}
    \right)
\]
Using \eqref{e:app-coad} the equations of motion corresponding to the Hamiltonian $h_0$ and the Poisson structure \eqref{e:app-KKS} on $\aut_0^*$ are therefore given by \eqref{e:lp1}-\eqref{e:lp3}.

Assume now that $n=3$ and that $K=S^1$. Thus the bracket $[.,.]_{\mathfrak{k}}$ is trivial.
Also, $\mu^M$ is the Euclidean inner product, whence $\mu^M u = u^{\flat}$ and $(\mu^M)^{-1}\Pi = \Pi^{\sharp}$ are the standard identifications for a vector field $u$ and a one-form $\Pi$ on $M$. 
We want to show that the system \eqref{e:lp1}-\eqref{e:lp3} is equivalent to \eqref{1e:EM_u_d}-\eqref{1e:cont_d}. Equation~\ref{1e:div0} follows because $u = [\mu^M]^{-1}\nu - [A^*]q\in\X_0(M)$ by construction of the isomorphism. Equation~\eqref{1e:cont_d} follows from \eqref{e:lp2} because $\rho^u q = \vv<u,\nabla q>$ and $\ad_{\mathfrak{k}}(X)^*q = 0$. Hence (neglecting the passage to the space of equivalence classes $\Om^1(M)/d\F$ from the second line onward) 
\begin{align*}
    [\mu^M]\dot{u}
    &= \dot{\nu} - \dot{q}[A]\\
    &=
    -L_u \mu^M u - L_u(q A) + q d i_u A - \by{1}{2}d\I^{-1}q^2 + \vv<u,\nabla q>A \\
    &=
    -\mu^M(\nabla_u u + \by{1}{2}\nabla\vv<u,u>) - q i_u dA - \by{1}{2}d\I^{-1}q^2 \\
    &=
    \mu^M\Big(
     -\nabla_u u + q u\times B - \nabla p 
     \Big) 
\end{align*}
where $i_u$ is the insertion (contraction) of a vector field into a $k$-form, we have used the Cartan formula $L_u A = di_u A + i_u dA$, $B = \textup{curl}\,(\mu^M)^{-1}A$ is the  vector field corresponding to the $2$-form  $dA$ via the Hodge-$*$ isomorphism, and $p$ is a function determined by the requirement that $\textup{div}(-\nabla_u u + q u\times B - \nabla p) = 0$. The equality $\mu^M(u\times B) = -i_u dA$ follows directly from a local coordinate calculation.


\begin{thebibliography}{99}

\bibitem{AS04}
A. Agrachev, Y. Sachkov, 
\emph{Control Theory from the Geometric Viewpoint},
Springer 2004. 


\bibitem{AS08}
A. Agrachev, Y. Sachkov, 
\emph{Solid Controllability in Fluid Dynamics}
In: Bardos C., Fursikov A. (eds), 
\emph{Instability in Models Connected with Fluid Flows I.}, 
International Mathematical Series, vol 6. Springer.
\href{https://doi.org/10.1007/978-0-387-75217-4_1}{https://doi.org/10.1007/978-0-387-75217-4\_1}

\bibitem{A66}
V.I. Arnold,
\emph{Sur la g\'eom\'etrie diff\'erentielle de groupes de Lie de dimension infinie et ses applications \`a l'hydrodynamique des fluides parfaits},
Annales de l'Institut Fourier \textbf{16} No. 1 (1966), p. 319-361.

\bibitem{AK98}
V. Arnold, B. Khesin,
\emph{Topological Methods in Hydrodynamics}, Springer 1998.



\bibitem{BLM01a}
A. Bloch, N. Leonard, J. Marsden,
\emph{Controlled Lagrangians and the Stabilization of
Mechanical Systems I: The First Matching Theorem},
IEEE Trans. on Sytems and Control, \textbf{45}, (2001), 2253-2270.

\bibitem{BLM01b}
A. Bloch, N. Leonard, J. Marsden,
\emph{Controlled Lagrangians and the stabilization of
Euler-Poincar\'e mechanical systems},
Int. J. Robust Nonlinear Control (2001) \textbf{11}:191-214.  

\bibitem{BKMS92}
A.M. Bloch, P.S. Krishnaprasad, J.E. Marsden, G. S\'anchez de Alvarez,
\emph{Stabilization of rigid body dynamics by internal and external torques},
Automatica \textbf{28}, Issue 4 (1992), Pages 745-756.

\bibitem{BMS97}
Bloch A.M., Marsden J.E., S\'anchez de Alvarez G. (1997), 
\emph{Feedback Stabilization of Relative Equilibria for Mechanical Systems with Symmetry.} In: Alber M., Hu B., Rosenthal J. (eds) \emph{Current and Future Directions in Applied Mathematics}. Birkh\"auser, Boston, MA.

 


\bibitem{C84}
F. Chen, 
\emph{Introduction to Plasma Physics and Controlled Fusion}, vol. 1, Plenum Press, New York, 1984

\bibitem{CJW00}
G.-Q. Chen, J.W. Jerome, D. Wang,
\emph{Compressible Euler-Maxwell equations},
Transport Theory and Statistical Physics \textbf{29}(3) (2000).
\href{https://doi.org/10.1080/00411450008205877}{https://doi.org/10.1080/00411450008205877}


  
\bibitem{CGG00}
S. Cordier, E. Grenier, Y. Guo,
\emph{Two-stream instabilities in plasmas},
Methods and applications of analysis \textbf{7}, Nr. 2 (2000), pp. 391-406.   
  
  





  
    

\bibitem{Davidson}
P.A. Davidson,
\emph{Introduction to Magnetohydrodynamics},
CUP 2017.







\bibitem{EM70}
    D. Ebin, J.E.  Marsden, \emph{Groups of diffeomorphisms and the motion of an incompressible fluid}, Ann. of Math. \textbf{92}(1) (1970), pp. 1037-1041.

\bibitem{FLL13}
    A. Futaki, H. Li, X. Li, \emph{On the first eigenvalue of the Witten-Laplacian and the diameter of compact shrinking Ricci solitons}, Annals of Global Analysis and Geometry \textbf{44}(2), pp. 105-114 (2013).

\bibitem{GBTV13}
F. Gay-Balmaz, C. Tronci, C. Vizman,
\emph{Geometric dynamics on the automorphism group of principal bundles:  Geodesic flows, dual pairs and  chromomorphism groups},
Journal of Geometric Mechanics
\textbf{5}(1), pp. 39-84 (2013).

\bibitem{GM14}
P. Germain, N. Masmoudi, 
\emph{Global existence for the Euler-Maxwell system},
Ann. Scient. \'Ec. Norm. Sup. \textbf{4} t. 47 (2014), p. 469-503.


\bibitem{GHK83}
J. Gibbons, D.D. Holm, B. Kupershmidt, 
\emph{The Hamiltonian structure of classical chromohydrodynamics}, Physica D \textbf{6} (1983), 179–194.





\bibitem{HMRW85}
D. Holm, J. Marsden, T. Ratiu, A. Weinstein,
\emph{Nonlinear stability of fluid and plasma equilibria}, 
Physics reports \textbf{123} (1-2), 1-116 (1985).










\bibitem{K85}
P.S. Krishnaprasad,
\emph{Lie-Poisson structures, dual-spin spacecraft and asymptotic stability},
Nonlinear Analysis: Theory, Methods \& Applications
Volume \textbf{9}, Issue 10, (1985), pp. 1011-1035. 
\href{https://doi.org/10.1016/0362-546X(85)90083-5}{https://doi.org/10.1016/0362-546X(85)90083-5}



\bibitem{Lin69}
J.G. Linhart, 
\emph{Plasma Physics}, Euratom 1969. 



\bibitem{MarsdenPCL}
Marsden, Jerrold E. (1999) \emph{Park City Lectures on Mechanics, Dynamics, and Symmetry.} In: \emph{Symplectic geometry and topology.} IAS/Park City Mathematics series. No.7. American Mathematical Society , pp. 335-430.

\bibitem{MEF}
  J. Marsden, D. Ebin, A. Fischer,
  \emph{Diffeomorphism groups, hydrodynamics and relativity},
  Proc. of the 13th Biennial Seminar of Canadian Mathematical Congress, (J. Vanstone, ed.), (1972), 135-279
  
\bibitem{MR-mas}
J. Matrsden, T. Ratiu,
\emph{Introduction to mechanics and symmetry}, Springer 1999. 

\bibitem{MWRSS83}
J. Marsden, A. Weinstein, T. Ratiu, R. Schmid, R. Spencer,
\emph{Hamiltonian systems with symmetry, coadjoint orbits and plasma physics},
Proc. IUTAM-IS1MM Symposium on Modern Developments in Analytical Mechanics (Torino 1982), Atti Acad. Sci. Torino Cl. Sci. Fis. Math. Natur. \textbf{117} (1983) pp. 289-340.

\bibitem{MRW84}
J. Marsden, T. Ratiu, A. Weinstein,
\emph{Semidirect products and reduction in mechanics},
Trans. Amer. Math. Soc. \textbf{281} Nr. 1 (1984).

  

\bibitem{MichorDG}
P. Michor,
\emph{Topics in differential geometry}, 
AMS Graduate Studies in Mathematics vol. 93, 2008.


\bibitem{M02}
R. Montgomery,
\emph{A Tour of Subrimannian Geometries, Their Geodesics and Applications},
AMS vol. 91, 2002. 


\bibitem{PB19}
M. Puiggal\'{i}, A. Bloch,
\emph{An Extension to the Theory of Controlled Lagrangians Using the Helmholtz Conditions},
J Nonlinear Sci (2019) \textbf{29}: 345-376.








\bibitem{WK92}
L. Wang, P.S. Krishnaprasad, 
\emph{Gyroscopic control and stabilization},
J Nonlinear Sci \textbf{2}, 367–415 (1992). \href{https://doi.org/10.1007/BF01209527}{https://doi.org/10.1007/BF01209527}




\end{thebibliography}
\end{document}